\documentclass[letterpaper, 10 pt, journal, onside, twocolumn, nofonttune]{IEEEtran}  % Comment this line out if you need a4paper
\IEEEoverridecommandlockouts                              % This command is only needed if 
\usepackage[top=0.833in, left=0.667in, right=0.667in, bottom=0.597in, includefoot]{geometry}

\usepackage{graphics} % for pdf, bitmapped graphics files
\usepackage{epsfig} % for postscript graphics files
\usepackage{color}
\usepackage[table,xcdraw,dvipsnames]{xcolor}

\usepackage{caption}
\usepackage{subcaption}
\usepackage{multirow}

\usepackage{amsmath} % assumes amsmath package installed
\usepackage{nicefrac}
\usepackage{amssymb}  % assumes amsmath package installed^\prime
\usepackage{bm}  % assumes amsmath package installed - for boldsymbol
\usepackage{stmaryrd} % double brackets https://tex.stackexchange.com/questions/107252/double-square-brackets/107464
\usepackage{mathtools}% http://ctan.org/pkg/mathtools

\usepackage{amsthm}
\usepackage{bussproofs}
\theoremstyle{plain}
\newtheorem{proposition}{Proposition}

\newtheorem{theorem}{Theorem}

\theoremstyle{definition}
\newtheorem{definitionx}{Definition} 
 \newenvironment{definition}
   {\pushQED{\qed}\definitionx}
   {\popQED\enddefinitionx}
\newtheorem{problemx}{Problem}
\newenvironment{problem}
  {\pushQED{\qed}\problemx}
  {\popQED\endproblemx}
\newtheorem{examplex}{Example}
\newenvironment{example}
  {\pushQED{\qed}\examplex}
  {\popQED\endexamplex}
\theoremstyle{remark}

\usepackage{listings,comment}
\usepackage{algorithm}
\usepackage[noend]{algorithmic}
\usepackage{tikz}
\usepackage{pgfplots}
\pgfplotsset{compat=1.15}

\usepackage{forest} % https://tex.stackexchange.com/questions/108710/square-edges-in-forest-package/108728#108728

\usetikzlibrary{arrows,positioning,chains,fit,shapes,automata} 
\usepgfplotslibrary{patchplots}
\tikzset{
    >=stealth',
    punkt/.style={
           rectangle,
           rounded corners,
           draw=black, very thick,
           text width=12em,
           minimum height=2em,
           text centered},
    pil/.style={
           ->,
           thick,
           shorten <=2pt,
           shorten >=2pt,},
	cross/.style={
    	cross out, 
        draw=black, 
        minimum size=2*(#1-\pgflinewidth), 
        inner sep=0pt, 
        outer sep=0pt},
	cross/.default={1pt}
}

\makeatletter
\let\NAT@parse\undefined
\makeatother

\usepackage{hyperref}  %hyperref still needs to be put at the end!
\usepackage[capitalise]{cleveref} %nice \ref command \cref, for example \cref{def:name} results into Definition 1
\crefname{examplex}{Example}{Examples}
\crefname{definitionx}{Definition}{Definitions}
\crefname{problemx}{Problem}{Problems}

\usepackage[symbols,toc]{glossaries-extra}

\makenoidxglossaries

\setabbreviationstyle[acronym]{long-short}
\preto\chapter{\glsresetall}

\newacronym{cegis}{CEGIS}{Counterexample-Guided Inductive Synthesis}
\newacronym{csp}{CSP}{Constraint Satisfiability Problem}
\newacronym{cp}{CP}{Constraint Programming}
\newacronym{smt}{SMT}{Satisfiability Modulo Theories}
\newacronym{lp}{LP}{Linear Programming}
\newacronym{milp}{MILP}{Mixed-Integer Linear Programming}
\newacronym{ips}{IPS}{Intelligent Physical System}
\newacronym{ltl}{LTL}{Linear Temporal Logic}
\newacronym{rtl}{RTL}{Temporal Logic over Reals}
\newacronym{stl}{STL}{Signal Temporal Logic}
\newacronym{mpc}{MPC}{Model Predictive Control}
\newacronym{itmp}{ITMP}{Integrated Task and Motion Planning}
\newacronym{ai}{AI}{Artificial Intelligence}
\newacronym{ff}{FF}{fast forward}
\newacronym{idtmp}{IDTMP}{iteratively deepened task and motion planning}
\newacronym{cosmop}{CoSMoP}{composition of safe motion primitives}
\newacronym{mld}{MLD}{mixed logical dynamic}
\newacronym{pomdp}{POMDP}{partially observable Markov decision process}
\newacronym{prstl}{PrSTL}{probabilistic signal temporal logic}
\newacronym{socp}{SOCP}{Second-Order Cone Programming}
\newacronym{rhc}{RHC}{receding horizon control}
\newacronym{kf}{KF}{Kalman filter}
\newacronym{ukf}{UKF}{unscented Kalman filter}
\newacronym{ekf}{EKF}{extended Kalman filter}
\newacronym{smc}{SMC}{sequencial Monte-Carlo}
\newacronym{lqr}{LQR}{Linear Quadratic Regulator}
\newacronym{lqg}{LQG}{Linear Quadratic Gaussian}
\newacronym{zoh}{ZOH}{zero order hold}
\newacronym{ir}{IR}{infrared}
\newacronym{cps}{CPS}{cyber-physical system}
\newacronym{dof}{DOF}{degrees of freedom}
\newacronym{rrt}{RRT}{Rapidly-exploring Random Tree}
\newacronym{ltlopt}{LTLOpt}{optimal control with linear temporal logic specifications}
\newacronym{ros}{ROS}{Robot Operating System}
\newacronym{bsc}{BSC}{Bounded Satisfiability Checking}
\newacronym{ompl}{OMPL}{Open Motion Planning Library}
\newacronym{dfa}{DFA}{deterministic finite automata}
\newacronym{dba}{DBA}{deterministic Büchi automata}
\newacronym{iis}{IIS}{Irreducibly Inconsistent Set}
\newacronym{dnf}{DNF}{Disjunctive Normal Form}
\newacronym{bmc}{BMC}{Bounded Model Checking}
\newacronym{idrtl}{idRTL}{iterative deepening Real-time Temporal Logic}
\newacronym{sat}{SAT}{Satisfiability}
\newacronym{mlo}{MLO}{Maximum Likelihood Observation}

\makenoidxglossaries
\usepackage{xifthen}% provides \isempty test
\usepackage{xspace}% provides an character blank space

\makeatletter
\newcommand{\pushright}[1]{\ifmeasuring@#1\else\omit\hfill$\displaystyle#1$\fi\ignorespaces}
\newcommand{\pushleft}[1]{\ifmeasuring@#1\else\omit$\displaystyle#1$\hfill\fi\ignorespaces}
\makeatother
\DeclareMathOperator*{\argmin}{arg\,min}
\DeclareMathOperator*{\argmax}{arg\,max}

\newcommand{\eye}[4]% size, x, y, rotation
{   \draw[rotate around={#4:(#2,#3)}] (#2,#3) -- ++(-.5*55:#1) (#2,#3) -- ++(.5*55:#1);
    \draw (#2,#3) ++(#4+55:.75*#1) arc (#4+55:#4-55:.75*#1);
    \draw[fill=gray] (#2,#3) ++(#4+55/3:.75*#1) arc (#4+180-55:#4+180+55:.28*#1);
    \draw[fill=black] (#2,#3) ++(#4+55/3:.75*#1) arc (#4+55/3:#4-55/3:.75*#1);
}

\newcommand{\mFormula}{\ensuremath{\varphi}}
\newcommand{\mTrue}{\ensuremath{\top}}

\newcommand{\mNot}{\ensuremath{\neg}}
\newcommand{\mAnd}{\ensuremath{\wedge}}
\newcommand{\mOr}{\ensuremath{\vee}}

\newcommand{\mUntil}{\ensuremath{\boldsymbol{U}}}
\newcommand{\mRelease}{\ensuremath{\boldsymbol{R}}}
\newcommand{\mAlways}{\ensuremath{\square}}
\newcommand{\mEventually}{\ensuremath{\Diamond}}

\newcommand{\mSat}{\ensuremath{\vDash}}

\title{ \LARGE \bf Active Perception and Control from PrSTL Specifications
\thanks{This work was supported in part by the National Science Foundation
under Grant IIS-1724070, Grant CNS-1830335 and Grant IIS-2007949.}
}

\author{Rafael Rodrigues da Silva$^1$ Vince Kurtz$^1$, and Hai Lin$^{1}$% <-this % stops a space
	\thanks{$^{1}$ All authors are with Department of Electrical Engineering, University of Notre Dame, Notre Dame, IN 46556, USA.
		{\tt\small (rrodri17@nd.edu;~vkurtz@nd.edu;~hlin1@nd.edu)}}
}

\begin{document}

\maketitle
\thispagestyle{empty}
\pagestyle{empty}

\begin{abstract}
Next-generation intelligent systems must plan and execute complex tasks with imperfect information about their environment. As a result, plans must also include actions to learn about the environment. This is known as active perception. Most active perception algorithms rely on reward or cost functions, which are usually challenging to specify and offer few theoretical guarantees. On the other hand, symbolic control methods can account for complex tasks using temporal logic but often do not deal well with uncertainties. This work combines symbolic control with active perception to achieve complex tasks in a partially observed and noisy control system with hybrid dynamics. Our basic idea is to employ a counterexample-guided-inductive-synthesis approach for control from probabilistic signal temporal logic (PrSTL) specifications. Our proposed algorithm combines bounded model checking (BMC) with sampling-based trajectory synthesis for uncertain hybrid systems. Active perception is inherently built into the framework because PrSTL formulas are defined in the chance domain. 
\end{abstract}

\section{Introduction}

%We have seen an unprecedented presence of intelligent systems in our daily life in the last decades. Many essential tasks are automated, from manufacturing processes to smart homes. However, their presence is usually limited in safety-critical tasks, where these automated systems must guarantee safe operation in the physical world. Currently, most intelligent systems have pre-programmed behaviors for a limited variety of contexts. Ideally, 

Future intelligent systems need to work reliably in uncertain and dynamic environments. Hence it is necessary for intelligent systems to perceive the world, extrapolate beliefs about the future, and take actions to improve confidence in these beliefs. This process is known as active perception.

Active perception has been pursued mainly in the robotic literature as a planning problem for a finite state \gls*{pomdp} \cite{valencia2013planning,agha2014firm,bai2014integrated,krishnamurthy2016partially}, which searches for a policy that maximizes the expected value of a reward function. However, specifying reward functions for complex, high-level tasks is difficult, and modeling the system as a finite state system can conceal critical system behaviors.

We are therefore motivated to specify high-level tasks as \gls*{prstl} formulas, which are straightforward to formulate as we illustrate in this paper. In addition, \gls*{prstl} combines real-time temporal logic with chance constraints, which fits our need to capture uncertainties in an active perception problem. We also model the system behavior as a switched linear system with Gaussian noise. This model allows us to inherit the computational efficiency and soundness of Kalman filtering. At the same time, it allows us to describe complex interactions of the system with its physical environment. Since we do not assume perfect knowledge of the initial state and the dynamics, the evolution of the state forms a random process that can be captured by a belief system. Then the active perception problem can be solved as a controller design problem for the belief system with respect to a given \gls*{prstl} specification. % with uncertain and differential constraints. 

% {\color{red} Active perception has been pursued mainly in the robotic literature, where methods like ... POMDP ... have been proposed ... (A brief review of active perception literature ...) However, ... 

% We are therefore motivated to ....} Here we do not assume perfect knowledge of the initial state and its dynamics. Thus, we represent the state as a random process called belief, and a stochastic dynamical system represents its dynamics. To solve this problem, we may need to minimize and control the uncertainty of the current belief using observations about the environment. 

%An active perception task involves both actions and perceptions. Thus, we propose to specify these tasks as \gls*{prstl} formulas, which combine real-time temporal logic with chance constraints. We also model the system behavior as a switched linear system with Gaussian noises. This model allows us to inherit the computational efficiency and soundness of Kalman filtering. At the same time, it helps us to represent complex behaviors of physical systems interacting with logical rules or controllers. Then we solve the active perception problem as a controller design for a given \gls*{prstl} specification with uncertain and differential constraints. 

%{\color{red} Our goal is to ... formal design ... scalable ... uncertainty  ... }  

Controller design for the belief system could quickly become a very daunting task due to the non-convexities of the \gls*{prstl} task specifications. In addition, the derived belief dynamics will introduce non-linear differential constraints. Hence, how to achieve a scalable design with respect to the size of specifications is our primary challenge. 
To mitigate this scalability challenge, our basic idea is to ease the computational cost by separating logical and dynamical constraints. %Hence, we do not restrict to a convex fragment of \gls*{prstl} formulas. To tackle the challenge of unknown observations, we approximate them using \gls*{mlo}\cite{platt2010belief}. Thus, we include the effects of the observation on the belief state during the planning. We introduce a hierarchical control synthesis approach that efficiently finds \gls*{prstl} controllers. 
This results in a hierarchical controller design. On the top, a custom \gls*{bmc} solver extracts sequences of convex chance constraints. At the bottom, a sampling-based motion planning searches for belief trajectories that satisfy these constraints. If no trajectory is found, we proposed an approach based on \gls*{cegis} \cite{alur2013syntax,reynolds2015counterexample} to update the \gls*{bmc}-generated constraints. 

An additional challenge is due to the fact that active perception requires measurement values that are not available during the planning. Some existing work on planning in the belief space uses a particular observation signal for planning, such as the \gls*{mlo} in \cite{platt2010belief}. Others consider observations as random variables. Such methods include the generalized belief state \cite{indelman2016towards} and covariance steering approaches \cite{ridderhof2020chance,zheng2021belief}. However, none of these approaches consider high-level task specifications. To tackle the challenge of unknown observations, we approximate them using the \gls*{mlo} approach \cite{platt2010belief}. As the name suggests, this approximation results in trajectories close to the one obtained when the observation is available. 

Although we deal with the active perception problem, our work is closely related to the existing \gls*{prstl} controller synthesis methods \cite{sadigh2015safe,dey2016fast,zhong2017fast}. In \cite{sadigh2015safe}, the authors extended the mixed-integer formulation used in \cite{raman2014model} to chance constraints and used a \gls*{rhc} approach to handle unknown measurements. Even though the approach focused only on a convex fragment of \gls*{prstl}, the computational cost motivated other works to search for more efficient solver techniques such as a faster \gls*{socp} optimization for safety properties \cite{zhong2017fast} and sampling-based optimization \cite{dey2016fast}. Distinct from these works, we consider the observation effects on the belief state update during planning and do not restrict ourselves to convex fragments of \gls*{prstl}.

Preliminary results of this work appeared in \cite{da2019active,rodriguesdasilva2021automatic,da2021symbolic}. In \cite{da2019active}, we  presented a solution for active perception and control from \gls*{prstl} formulas. In this work, we extend the approximated belief system introduced in \cite{da2019active} to a switched linear system with Gaussian noise. In \cite{rodriguesdasilva2021automatic,da2021symbolic}, \gls*{bmc} and \gls*{lp} were combined to achieve a scalable design for deterministic systems with respect to complex specifications in real-time temporal logic. In this paper, we harvest the insights gained from our prior work on the deterministic system design and propose a new sampling-based method for active perception that uses \gls*{lp} and quantitative semantics of \gls*{prstl} formulas. Code is available at \texttt{\url{https://codeocean.com/capsule/0013534/tree}}.

This paper is organized as follows. In \cref{sec:pr_preliminaries}, we introduce the main concepts used in the paper. In \cref{sec:overall_approach}, we overview the proposed approach and, in \cref{sec:det,sec:abs,sec:bmc,sec:feas,sec:idprstl}, we describe the main steps of this method in more detail. In \cref{sec:experiments}, we illustrate our approach with two examples. Finally, \cref{sec:conclusion} concludes the work.

\section{Preliminaries}\label{sec:pr_preliminaries}

\subsection{Polytopes}

A \emph{polytope} $\mathcal{X} \subseteq \mathbb{R}^n$ is a set in $\mathbb{R}^n$ defined by the intersection of a finite number of closed half-spaces, i.e, $\mathcal{X} := \cap_i \{ \mu_i(\boldsymbol{x}) \leq 0 \}$, where $\mu_i(\boldsymbol{x}) := \boldsymbol{h}_i^\intercal \boldsymbol{x} + c_i$ is a linear function, $H(\mathcal{X})$ is the set of linear functions that defines the polytope $\mathcal{X}$, $\boldsymbol{h}_i \in \mathbb{R}^n$ and $c_i \in \mathbb{R}$ are constants. We can also represent a compact polytope $\mathcal{X} \subset \mathbb{R}^n$ as the convex hull of its vertices, i.e., $\mathcal{X} = \text{conv}\big(V\big)$, where $V$ is a set of vertices.

\subsection{System}\label{sec:prsystem}

We consider switched linear control systems as follows:
\begin{equation}\label{eq:prsystem}
    \begin{aligned}
	    \boldsymbol{x}_{k+1} = & A_{q} \boldsymbol{x}_k + B_{q} \boldsymbol{u}_k + W_{q} \boldsymbol{w}_k, & \boldsymbol{w}_k \sim \mathcal{N}(0, I)\\
	    \boldsymbol{y}_k = & C_{q} \boldsymbol{x}_k + n_{q}(\boldsymbol{x}_k) \boldsymbol{v}_k, & \boldsymbol{v}_k \sim \mathcal{N}(0, I),
    \end{aligned}
\end{equation}
where $\boldsymbol{x}_k \in \mathbb{R}^n$ are the state variables, $\boldsymbol{u}_k \in \mathcal{U} \subseteq \mathbb{R}^m$ are the input variables, $\mathcal{U} \subseteq \mathbb{R}^m$ is a polytope, $\boldsymbol{y}_k \in \mathbb{R}^p$ are the output variables. Each system mode (command) $q \in Q = \{1, 2, \dots, N \}$ is defined by a noise function $n_q : \mathbb{R}^n \rightarrow \mathbb{R}^n$, and constant matrices $A_q \in \mathbb{R}^{n \times n}$, $B_q \in \mathbb{R}^{n \times m}$, $W_q \in \mathbb{R}^{n \times n}$, and $C_q \in \mathbb{R}^{p \times n}$. We assume that the system is subject to mutually uncorrelated zero-mean stationary Gaussian additive disturbances $\boldsymbol{v}_k \in \mathcal{N}(0, I_n)$ and $\boldsymbol{w}_k \in \mathcal{N}(0, I_p)$, where $I_n$ is the identity matrix with dimension $n$. Note that this dynamical system can arise from linearization and sampling of a more general continuous system. In such a case, we denote the sampling period as $T_s$, where $T_s = t_{k+1} - t_k $ for all $k \in \mathbb{N}_{\geq 0}$. 
%We assume that the uncertainty is stable, meaning that the uncertainty does not increase infinitely over time.

\begin{example}\label{ex:light-dark-system}
     Let us consider the scenario presented in \cite[Sec. VI.A]{platt2010belief}, called the light-dark domain. In this example, we have a motion planning problem where the robot's position in the workspace is uncertain. Additionally, the measurment noise depends on the robot's position, as shown in \cref{fig:light-dark-example-workspace}. 
     
     The state space is the workspace plane, $\boldsymbol{x} \in \mathbb{R}^2$, and the robot is modeled as a first-order system controlled by velocity, $\boldsymbol{u} \in \mathbb{R}^2$: $\boldsymbol{x}_{k+1} = \boldsymbol{x}_k + 0.25 \boldsymbol{u}_k$. The observation function is $\boldsymbol{y}_k = \boldsymbol{x}_k + n(\boldsymbol{x}_k) \boldsymbol{v}_k$ with a zero-mean Gaussian noise function as follows: 
     \begin{equation}
         n(\boldsymbol{x}) = 0.1(5 - x_1)^2 + const,
     \end{equation}
     where $\boldsymbol{x} = [x_1, x_2]^\intercal$. 
     
     \begin{figure}
        \centering
        \includegraphics[width=0.9\linewidth]{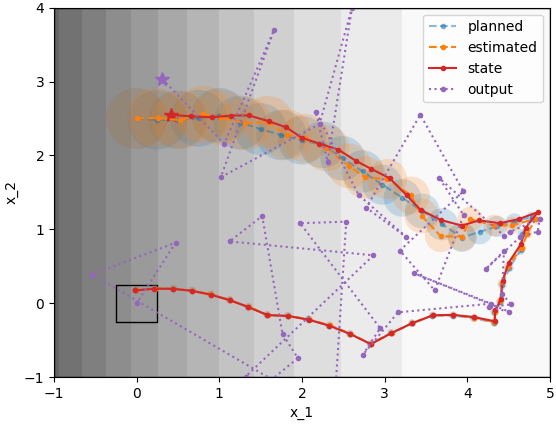}
        \caption{Light-dark example. The shade in the workspace is proportional to measurement noise at that position, the black box is the target, and the blue and orange regions are the $95\%$ confidence region of the belief. The estimated belief trajectory (in orange) tracks the planned belief trajectory (in blue).}
        \label{fig:light-dark-example-workspace}
        \vspace{-0.5cm}
    \end{figure}
    
    The initial belief is an isotropic Gaussian distribution centered at position $[0, 0]^\intercal$ with covariance $0.1I_2$. As shown in \cref{fig:light-dark-example-workspace}, the uncertainty of the initial belief does not allow the robot to achieve $95\%$ of confidence to be inside the target region without getting more information. Hence, the robot must move to the light region to improve its position estimate before approaching the goal. 
\end{example}

\begin{example}\label{ex:laser-grasp-system}
    Now, consider a planar robot manipulator that must grasp a round puck and place it down in a target location as illustrated in \cref{fig:laser-grasp-example-robot}. However, the puck and target locations are not known perfectly a priori. Thus, we use a camera coupled at the end effector to measure the relative distance. The challenge is that the camera is obstructed by the puck when the robot is holding it. Therefore, only trajectories that ``learn'' the target location before grasping the puck can satisfy the specification.
    
  \begin{figure}
    \begin{subfigure}[b]{0.24\textwidth}
        \centering
        \begin{tikzpicture}[auto, >=latex']
            \draw[-] (-0.5, 0) -- (0.5, 0);
            \draw[-] (-0.5, 0) -- (-0.14644661, -0.35355339);
            \draw[-] (-0.25, 0) -- (0.10355339, -0.35355339);
            \draw[-] (0, 0) -- (0.35355339, -0.35355339);
            \draw[-] (0.25, 0) -- (0.60355339, -0.35355339);
            \draw[-] (0.5, 0) -- (0.85355339, -0.35355339);
            \draw[-] (0, 0) -- (0.5, 0.866) circle (2pt) -- (0, 1.732) circle (2pt) -- (0, 1.982);
            \draw[-] (-0.125, 2.232) -- (-0.125, 1.982) -- (0.125, 1.982) -- (0.125, 2.232);
            \draw[-, dashed] (0, 2.107) -- (0.70710678, 2.81410678);
            \draw[-, dashed] (0, 2.107) -- (-0.70710678, 2.81410678);
            \node at (0, 3) {camera};
            \filldraw[black] (-1.9,0.8660254) circle (2pt) node[anchor=east] {puck};
            \draw[-] (0, 0) -- (-0.5      ,  0.8660254) circle (2pt) -- (-1.5      ,  0.8660254) circle (2pt) -- (-1.75     ,  0.8660254);
            \draw[-] (-2, 0.7410254) -- (-1.75, 0.7410254) -- (-1.75, 0.9910254) -- (-2, 0.9910254);
            \draw[-] (0, 0) -- (-0.25881905,  0.96592583) circle (2pt) -- (-0.96592583,  0.25881905) circle (2pt) -- (-0.96592583,  0.00881905);
            \draw[-] (-0.84092583, -0.24118095) -- (-0.84092583, 0.00881905) -- (-1.09092583, 0.00881905) -- (-1.09092583, -0.24118095);
            \draw[-] (-0.59092583, 0) -- (-0.59092583, -0.5) -- (-1.34092583, -0.5)  node[anchor=east] {target} -- (-1.34092583, 0);
        \end{tikzpicture}
        \caption{Planar robot manipulator}
        \label{fig:laser-grasp-example-robot}
    \end{subfigure}
    \begin{subfigure}[b]{0.22\textwidth}
        \centering
        \includegraphics[width=\textwidth]{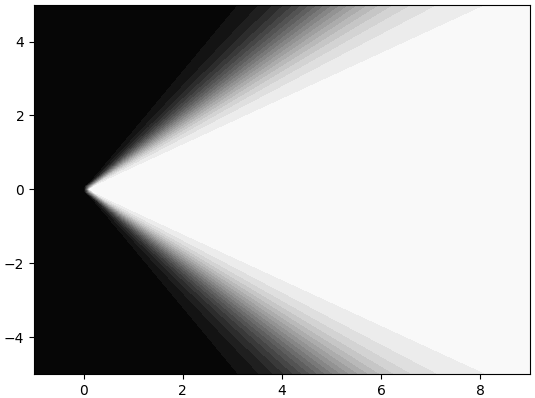}
        \caption{noise model with camera at origin $(0,0)$ and $\theta = 0$ about different relative distances}
        \label{fig:laser-grasp-example-noise-model}
    \end{subfigure}
    \caption{A camera is coupled at the robot end effector. The objective is to move the puck inside the target region. However, the puck and target location are not perfectly known at the beginning of the task. Hence, the camera is used to measure the relative distance between the end effector and these targets.}
    \label{fig:laser-grasp-example}
    \vspace{-0.5cm}
 \end{figure}
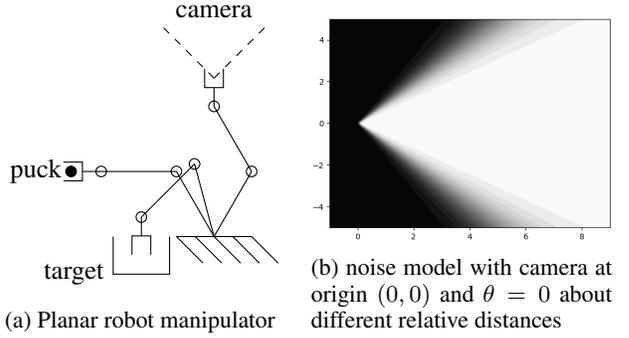

    The system state space is formed by three variables: $\boldsymbol{p}_{puck}, \boldsymbol{p}_{target} \in \mathbb{R}^2$ and $\theta \in SO(2)$, where $\theta$ is the end effector orientation angle, and $\boldsymbol{p}_{puck}$ and $\boldsymbol{p}_{target}$ are the relative position between the end effector and the puck and target positions. The dynamics are a first-order system controlled by linear velocities $\boldsymbol{v} \in \mathbb{R}^2: \|\boldsymbol{v}\|_{\infty} \leq 1$ and angular velocities $\omega \in SO(2) : |\omega| \leq 45^o$, and a three ($N = 3$) system commands: gripper opened ($q = 1$), gripper closing/opening ($q = 2$), and gripper closed ($q = 3$):
    \begin{equation}
    \begin{aligned}
         \boldsymbol{p}_{puck, k + 1}  = &  
         \begin{cases}
            \boldsymbol{p}_{puck, k} + \boldsymbol{v}_k, & \textbf{if } q = 1, \\
            \boldsymbol{p}_{puck, k}, & \textbf{if } q = 2 \textbf{ or } q = 3,
         \end{cases} \\
         \boldsymbol{p}_{target, k + 1}  = &  
         \begin{cases}
            \boldsymbol{p}_{target, k} + \boldsymbol{v}_k, & \textbf{if } q = 1 \textbf{ or } q = 3, \\
            \boldsymbol{p}_{target, k}, & \textbf{if } q = 2,
         \end{cases} \\
         \theta_{k + 1}  = &  
         \begin{cases}
            \theta_{k} + \omega_k, & \textbf{if } q = 1 \textbf{ or } q = 3, \\
            \theta_{k}, & \textbf{if } q = 2.
         \end{cases}
    \end{aligned}
    \end{equation}
    Intuitively, we assume that the robot is holding the puck if the gripper is closed. Hence, the relative distance to the puck only changes if the gripper is opened. Additionally, we assume no process noise (i.e. $W_k = O_{3, 3}$).
    
    The measured (output) variables are the puck $\boldsymbol{p}_{puck, k}^y$ and target $\boldsymbol{p}_{target, k}^y$ relative distances. However, if the robot is holding the puck, we cannot see the target location anymore. Thus, the noise is modeled as follows:
    \begin{equation}
    \begin{aligned}
         \boldsymbol{p}_{puck, k}^y  = &  
         \begin{cases}
            \boldsymbol{p}_{puck, k} + n(\boldsymbol{p}_{puck, k}, \theta_k) I_2 \boldsymbol{\nu}_{k}, & \textbf{if } q = 1, \\
            \boldsymbol{p}_{puck, k} + v_{large} I_2 \boldsymbol{\nu}_{k}, & \textbf{otherwise},
         \end{cases} \\
         \boldsymbol{p}_{target, k}^y  = &  
         \begin{cases}
            \boldsymbol{p}_{target, k} + n(\boldsymbol{p}_{target, k}, \theta_k) I_2\boldsymbol{\nu}_{k}, & \textbf{if } q = 1, \\
            \boldsymbol{p}_{target, k} + v_{large} I_2 \boldsymbol{\nu}_{k}, & \textbf{otherwise},
         \end{cases}
    \end{aligned}
    \end{equation}
    where $I_2 \in \mathbb{R}^{2 \times 2}$ is an identity matrix, $v_{large} \in \mathbb{R}$ is a large constant scalar, and the noise function $n: \mathbb{R}^2 \times SO(2) \rightarrow \mathbb{R}$ is modeled as a symmetric squashing function about the angle difference between the relative position orientation ($\angle \boldsymbol{p}_k$) and the end effector orientation $\theta$, as shown in \cref{fig:laser-grasp-example-noise-model}.
\end{example}

\subsection{Belief System}\label{sec:belief_state}

Since only noisy observations $\boldsymbol{y}_k$ are available in System (\ref{eq:prsystem}), we must estimate the state $\boldsymbol{x}_k$. To do so, we assume that the controller keeps track of a \textit{history} of observations and actions up to the current time instant. History can be compactly represented as a random process $\boldsymbol{X}_k$ over the state space such that the \textit{belief state} $P(\boldsymbol{X}_k = \boldsymbol{x}_k)$ is a sufficient statistic for history \cite{krishnamurthy2016partially}. We denote $P(\boldsymbol{X}_k = \boldsymbol{x}_k)$ by $P(\boldsymbol{x}_k)$ and $P(\boldsymbol{Y}_k = \boldsymbol{y}_k)$ by $P(\boldsymbol{y}_k)$, where $\boldsymbol{Y}_k$ is the output random process. The belief state can be tracked using a Bayesian filter:
\begin{equation*}
P(\boldsymbol{x}_{k+1}) = \eta P(\boldsymbol{y}_{k+1}|\boldsymbol{x}_{k+1},\boldsymbol{u}_k)\int_{\boldsymbol{x}} P(\boldsymbol{x}_{k+1}|\boldsymbol{x},\boldsymbol{u}_k)P(\boldsymbol{x}_k) dx,
\end{equation*}
where $\eta$ is a normalization constant \cite{krishnamurthy2016partially}. 

System (\ref{eq:prsystem}) is subject to additive Gaussian noise. Hence, the belief state is a Gaussian process with state dynamics determined by a Kalman filter (KF)\footnote{Note that we could use a different noise model, but the Bayesian filter would no longer be a KF.} \cite{chui2017kalman}. Thus, the belief state $\boldsymbol{X}_k \sim  \mathcal{N}(\hat{\boldsymbol{x}}_k,\Sigma_k^x)$ is a random process with Gaussian distribution. Given a system command consisting of a mode $q$ and an input $\boldsymbol{u}$, and an observation $\boldsymbol{y}_{k+1} \in \mathbb{R}^p$, we can represent the belief dynamics by a deterministic dynamical system defined over the belief mean $\hat{\boldsymbol{x}}_k \in \mathbb{R}^{n}$ and covariance $\Sigma_k^x \in \mathbb{R}^{n \times n}$:
\begin{equation}\label{eq:beliefsys}
\begin{aligned}
\hat{\boldsymbol{x}}_{k+1} = & (I_n - L_{q,k} C_{q}) (A_{q} \hat{\boldsymbol{x}}_k + B_{q} \boldsymbol{u}_k) + L_{q,k} \boldsymbol{y}_{k + 1}, \\
\Sigma_{k+1}^x = & (I_n - L_{q,k} C_{q}) \Sigma_{q, k}^{x,+},
\end{aligned}
\end{equation}
where $\Sigma_{q, k}^{x,+} = A_{q} \Sigma_k^x A_{q}^\intercal  + W_{q} W_{q}^\intercal$ is the a priori estimate covariance, $L_{q,k} = \Sigma_k^{x,+} C_{q}^\intercal (\Sigma_k^z)^{-1}$ is the Kalman gain, and $\Sigma_{q, k}^{z,+} = C_{q} \Sigma_{q, k}^{x,+} C_{q}^\intercal  + V_{q} V_{q}^\intercal$ is the innovation covariance. We assume that the noise function $n_{q}$ can be approximated as $V_{q} = n_{q}(\hat{\boldsymbol{x}}_k)$. %Note that this system for the mean and covariance is deterministic, but represents a random process.

Although System (\ref{eq:beliefsys}) is deterministic (i.e., an a priori mean $\hat{\boldsymbol{x}}_k$ and covariance $\hat{\Sigma}_k^x$, a command $(q, \boldsymbol{u})$ and an observation $\boldsymbol{y}_{k+1}$ define an unique a posteriori mean $\hat{\boldsymbol{x}}_{k+1}$ and covariance $\hat{\Sigma}_{k+1}^x$), the underlying system is probabilistic (i.e., the mean $\hat{\boldsymbol{x}}_k$ and covariance $\hat{\Sigma}_k^x$ represent a Gaussian random variable $\boldsymbol{X}_k$). A belief trajectory $\boldsymbol{\beta}$ is defined as a sequence $\boldsymbol{X}_0 \xrightarrow{q_0, \boldsymbol{u}_0, \boldsymbol{y}_1} \boldsymbol{X}_1 \dots$. A transition $\boldsymbol{X}_k\xrightarrow{q_k, \boldsymbol{u}_k, \boldsymbol{y}_k}\boldsymbol{X}_{k+1}$ represents the process of applying a command $q_k$ and input $\boldsymbol{u}_k$ at instant $k$ and waiting for an observation $\boldsymbol{y}_{k+1}$ at instant $k + 1$ to update the next belief state $\boldsymbol{X}_{k+1}$. We denote a prefix of a trajectory $\boldsymbol{\beta}$ by $\boldsymbol{\beta}_K = Prefix_K(\boldsymbol{\beta}) = \boldsymbol{X}_0 \xrightarrow{q_0, \boldsymbol{u}_0, \boldsymbol{y}_1} \boldsymbol{X}_1 \dots \xrightarrow{q_{K-1}, \boldsymbol{u}_{K-1}, \boldsymbol{y}_K} \boldsymbol{X}_K$.
 
 \begin{example}
     Consider the light-dark domain presented in \cref{ex:light-dark-system}. An example of a trajectory is shown in \cref{fig:light-dark-example-workspace}. The state trajectory (in red) is unknown to the system. The robot can only observe the noisy output (in purple). However, the KF allows us to use these noisy observations to reduce the uncertainty in the belief trajectory (in orange). 
 \end{example}
 
% \begin{remark}
% Note that we assume deterministic system commands $e_k \in E$ and inputs $\boldsymbol{u}_k \in \mathcal{U}$ at each instant $k$, and that the noise function $n_{e_k}$ can be approximated as $V_{e_k} = n_{e_k}(\hat{\boldsymbol{x}}_k)$. Therefore, the system is equivalent to a time variant linear system. 
% \end{remark}

\subsection{Belief Cones}\label{sec:belief_cones}

Polytopes in the original state space correspond to cones in the belief space. These cones $\mathcal{B} \subseteq \mathbb{R}^{n (n + 1)}$ can be characterized as the intersection of a finite set of second order cones in the (Gaussian) parameter space. Intuitively, any Gaussian distribution with a mean $\hat{\boldsymbol{x}} \in \mathbb{R}^n$ and covariance $\Sigma^x \in \mathbb{R}^{n \times n}$ inside this cone (i.e., $(\hat{\boldsymbol{x}}, \Sigma^x) \in \mathcal{B}$) satisfies a set of probabilistic linear predicates (i.e, $\bigwedge_i P(\mu_i(\boldsymbol{x}) \leq 0) \geq 1 - \epsilon_i$). For simplicity, we will denote that a Gaussian random variable $\boldsymbol{X}_k \sim \mathcal{N}(\hat{x}_k, \Sigma_k^x)$ satisfies a probabilistic linear predicate $P(\mu_i(\boldsymbol{x}) \leq 0) \geq 1 - \epsilon_i$ by $\boldsymbol{X}_k \mSat P(\mu_i(\boldsymbol{x}) \leq 0) \geq 1 - \epsilon_i$. Therefore, a belief cone is defined as:
\begin{equation}
\begin{aligned}
     \mathcal{B} := & \{ \boldsymbol{b} \in \mathbb{R}^{n (n+1)}: \bigwedge_i \boldsymbol{X} \mSat P(\boldsymbol{h}_i^\intercal \hat{\boldsymbol{x}}_k + c_i \leq 0) \geq 1 - \epsilon_i\} \\
     := & \cap_i \{ \boldsymbol{h}_i^\intercal \hat{\boldsymbol{x}}_k + c_i + \Phi^{-1}(1 - \epsilon_i) \sqrt{\boldsymbol{h}_i^\intercal \Sigma_k^x \boldsymbol{h}_i}  \leq 0 \},
\end{aligned}
\end{equation}
where $\boldsymbol{b} \in \mathbb{R}^{n (n+1)}$ is the Gaussian distribution parameter variable, $\boldsymbol{h}_i^\intercal \in \mathbb{R}^n$, $c_i \in \mathbb{R}$ and $\epsilon_i \in [0, 0.5]$ are constants, $\Phi(v)$ and  $\Phi^{-1}(p)$ are the cumulative distribution and quantile functions of the standard Gaussian distribution $V \sim \mathcal{N}(0,1)$, i.e., $\Phi(v) = P(V \leq v)$ and $\Phi^{-1}(p) \leq v$ if and only if $p \leq \Phi(v)$. 
%We can easily see that $\boldsymbol{X} \mSat P(\boldsymbol{h}_i^\intercal \hat{\boldsymbol{x}}_k + c_i \leq 0) \geq 1 - \epsilon_i$ if and only if $\boldsymbol{h}_i^\intercal \hat{\boldsymbol{x}}_k + c_i + \Phi^{-1}(1 - \epsilon_i) \sqrt{\boldsymbol{h}_i^\intercal \Sigma_k^x \boldsymbol{h}_i}  \leq 0$ from Gaussian distribution properties such as linear transformation and the quantile function definition.

\subsection{Probabilistic Signal Temporal Logic}\label{sec:prstl}

We specify the requirements of a system belief trajectory using 
\gls*{prstl} formulas. These formulas are defined recursively according to the following grammar:
\begin{align*}
    \phi := & \pi^\mu_\epsilon | \pi^\mathbb{Q} | \phi_1 \mAnd \phi_2 \\
    \mFormula := & \phi | \mFormula_1 \mAnd \mFormula_2 | \mFormula_1 \mOr \mFormula_2 |  \mFormula_1 \mUntil_{[a,b]} \mFormula_2 |  \mAlways_{[a,b]} \mFormula,
\end{align*}
where $\pi$ is a predicate, $\mFormula$, $\mFormula_1$, and $\mFormula_2$ are \gls*{prstl} formulas, and $\phi$, $\phi_1$, and $\phi_2$ are  \gls*{prstl} state formulas. Predicates can be one of two types: atomic and probabilistic. An atomic predicate $\pi^\mathbb{Q}$ is a statement about the system modes (commands) and is defined by a set $\mathbb{Q} \subseteq Q$ of modes. A probabilistic predicate $\pi^\mu_\epsilon$ is a statement about the belief $\boldsymbol{X}_k$ defined by a linear function $\mu : \mathbb{R}^n \rightarrow \mathbb{R}$ and a a tolerance $\epsilon \in [0,0.5]$. The operators $\mAnd, \mOr$ are Boolean operators conjuntion and disjunction. The temporal operators $\mUntil$ and $\mAlways$ stand for until and always. In \gls*{prstl}, these operators are defined by an interval $[a,b] \subseteq \mathbb{N}_{\geq 0}$.

We denote the fact that a belief trajectory $\boldsymbol{\beta}$ satisfies an \gls*{prstl} formula $\mFormula$ with $\boldsymbol{\beta} \mSat \mFormula$. Furthermore, we write $\boldsymbol{\beta} \mSat_k \mFormula$ if the trajectory $\boldsymbol{X}_k \xrightarrow{q_k, \boldsymbol{u}_k, \boldsymbol{y}_{k+1}} \boldsymbol{X}_{k + 1} \dots$ satisfies $\mFormula$.  Formally, the following semantics define the validity of a formula $\mFormula$ with respect to the trajectory $\boldsymbol{\beta}$:
\begin{itemize}
  \item $\boldsymbol{\beta} \mSat_k \pi^\mathbb{Q}$ if and only if $k = 0$ or $q_{k-1} \in \mathbb{Q}$,
  \item $\boldsymbol{\beta} \mSat_k \pi^\mu_\epsilon$ if and only if $p\big(\mu(\boldsymbol{x}_k) \leq 0\big) \geq 1 - \epsilon$,
  \item $\boldsymbol{\beta} \mSat_k \mFormula_1 \mAnd \mFormula_2$ if and only if $\boldsymbol{\beta} \mSat_k \mFormula_1$ and $\boldsymbol{\beta} \mSat_k \mFormula_2$,
  \item $\boldsymbol{\beta} \mSat_k \mFormula_1 \mOr \mFormula_2$ if and only if $\boldsymbol{\beta} \mSat_k \mFormula_1$ or $\boldsymbol{\beta} \mSat_k \mFormula_2$,
    \item $\boldsymbol{\beta} \mSat_{k} \mFormula_1 \mUntil_{[a,b]} \mFormula_2$ if and only if $\exists k^\prime \in [k + a, k + b]$ s.t. $\boldsymbol{\beta} \mSat_{{k^\prime}}\mFormula_2$, and, $\forall k^{\prime\prime} \in [k + a, k^\prime]$, $\boldsymbol{\beta} \mSat_{{k^{\prime\prime}}}\mFormula_1$;
    \item $\boldsymbol{\beta} \mSat_{k} \mAlways_{[a,b]} \mFormula$ if and only if $\forall k^\prime \in [k + a, k + b]$, $\boldsymbol{\beta} \mSat_{{k^\prime}}\mFormula_2$,
  \item $\boldsymbol{\beta} \mSat \mFormula$ if and only if $\boldsymbol{\beta} \mSat_0 \mFormula$,
\end{itemize} 
where the temporal operators are indexed by a delay $a \in \mathbb{N}_{\geq 0}$ and a deadline $b \in \mathbb{N}_{\geq 0}: a < b \leq \infty$. We can derive other operators such as \emph{true} ($\top = \pi^Q$), \emph{false} ($\perp = \pi^\emptyset$), \emph{release} ($\mFormula_1 \mRelease_{[a,b]} \mFormula_2 = (\mFormula_2 \mUntil_{[a,b]} \mFormula_1) \mOr \mAlways_{[a,b]} \mFormula_2$) and eventually ($\mEventually_{[a,b]} \mFormula = \top \mUntil_{[a, b]} \mFormula$).

 \begin{example}\label{ex:light-dark-specs}
     We can specify motion planning problems as \gls*{prstl} formulas. Consider the light-dark example, the specification is:
     \begin{equation}
     \mFormula = safe \boldsymbol{U}_{[0,240]} \mAlways_{[0,10]} target,
     \end{equation}
     where $safe = \pi_{0.01}^{-x_1 - 1} \mAnd \pi_{0.01}^{x_1 - 5} \mAnd \pi_{0.01}^{-x_2 - 1} \mAnd \pi_{0.01}^{x_2 - 4}$ and $target = \pi_{0.05}^{-x_1 - 0.25} \mAnd \pi_{0.05}^{x_1 - 0.25} \mAnd \pi_{0.05}^{-x_2 - 0.25} \mAnd \pi_{0.05}^{x_2 - 0.25}$.
     In plain English, the robot must satisfy each safety boundary with $99\%$ confidence until it achieves each target boundary with $95\%$ confidence within $240$ time instants and stays in the target for $10$ time instants.
 \end{example}

 \begin{example}\label{ex:laser-grasp-specs}
     %We can specify robot manipulation planning problem as \gls*{prstl} formulas. 
     Consider the planar robot manipulator example, the specification is:
     \begin{align*}
        \mFormula = & \mFormula_{step,1}(\mFormula_{step,2}(\mFormula_{step,3}(\mFormula_{step,4}))) \\
        \mFormula_{step,1}(\mFormula^\prime) = & \big((\pi^{\{1\}} \mAnd safe) \mOr (\pi^{\{1\}} \mAnd approach) \\ 
        & \hspace{2cm} \mUntil_{[0, \infty]} (\pi^{\{1\}} \mAnd approach) \mAnd \mFormula^\prime\big) \\
        \mFormula_{step,2}(\mFormula^\prime) = & \big((\pi^{\{1\}} \mAnd approach) \mOr (\pi^{\{2\}} \mAnd grasp)  \\ 
        & \hspace{2cm} \mUntil_{[0, \infty]} (\pi^{\{2\}} \mAnd grasp) \mAnd \mFormula^\prime\big) \\
        \mFormula_{step,3}(\mFormula^\prime) = & \big((\pi^{\{2\}} \mAnd grasp) \mOr (\pi^{\{3\}} \mAnd move)  \\ 
        & \hspace{2cm} \mUntil_{[0, \infty]} (\pi^{\{3\}} \mAnd move) \mAnd \mFormula^\prime\big) \\
        \mFormula_{step,4}(\mFormula^\prime) = & \big((\pi^{\{3\}} \mAnd move) \mOr (\pi^{\{2\}} \mAnd place\_down)  \\ 
        & \hspace{2cm} \mUntil_{[0, \infty]} (\pi^{\{2\}} \mAnd place\_down)\big),
     \end{align*}
     where the propositions are formed by the predicates illustrated in \cref{fig:laser-grasp-specs} with a tolerance of $0.05$. 
     The robot must move the gripper without colliding with the puck until it reaches the grasp configuration and, next, move the puck to the target. 
    \begin{figure}
        \centering
        \begin{subfigure}[b]{0.24\textwidth}
            \centering
            \includegraphics[width=\textwidth]{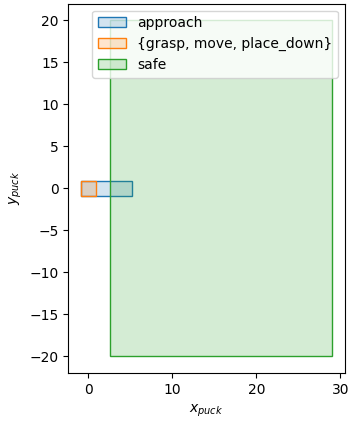}
            \caption{Puck}
        \end{subfigure}
        \begin{subfigure}[b]{0.24\textwidth}
            \centering
            \includegraphics[width=\textwidth]{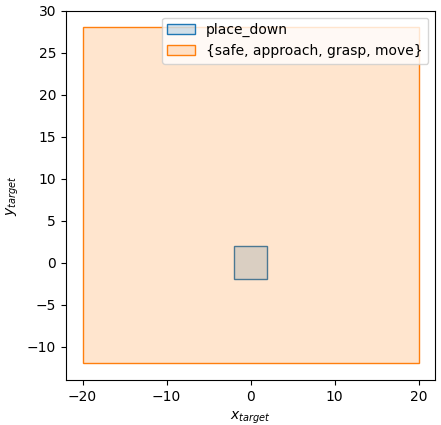}
            \caption{Target}
        \end{subfigure}
        \begin{subfigure}[b]{0.3\textwidth}
            \centering
            \includegraphics[width=\textwidth]{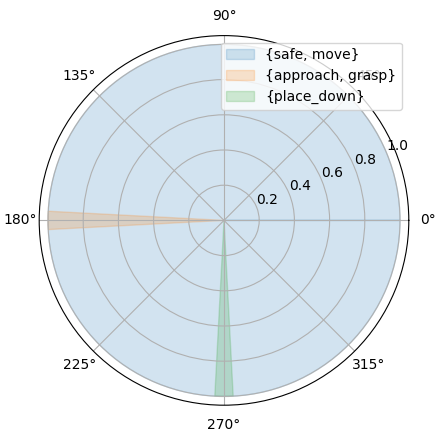}
            \caption{Orientation}
        \end{subfigure}
        \caption{Illustrative representation of the formula in \cref{ex:laser-grasp-specs}.}
        \label{fig:laser-grasp-specs}
    \vspace{-0.5cm}
    \end{figure}
 \end{example}

We denote the closure $cl(\mFormula)$ of a \gls*{prstl} formula $\mFormula$ as the smallest set satisfying $\mFormula_1 \circ \mFormula_2 \in cl(\mFormula)$ (or $\mAlways_{[a,b]}\mFormula^\prime \in cl(\mFormula)$), then $\mFormula_1, \mFormula_2 \in cl(\mFormula)$ (or $\mFormula^\prime \in cl(\mFormula)$), where $\circ \in \{\mOr,\mAnd, \mUntil_{[a,b]} \}$. 

\begin{example}
    Consider the formula $\mFormula$ in \cref{ex:light-dark-specs}. The main formula $\mFormula$ is in the closure $\mFormula \in cl(\mFormula)$. All of its atomic predicates are also in its closure (e.g., $\pi_{0.01}^{-x_1 - 1} \in cl(\mFormula)$). Similarly, any sub-formula in $\mFormula$ is also in its closure (e.g., $\pi_{0.01}^{-x_1 - 1} \mAnd \pi_{0.01}^{x_1 - 5} \in cl(\mFormula)$ and $\mAlways_{[0,10]}  \pi_{0.05}^{-x_1 - 0.25} \mAnd \pi_{0.05}^{x_1 - 0.25} \mAnd \pi_{0.05}^{-x_2 - 0.25} \mAnd \pi_{0.05}^{x_2 - 0.25} \in cl(\mFormula)$).
\end{example}

\subsection{Quantitative Semantics for PrSTL}

We can define a real-valued function $\rho^\mFormula$ of the trajectory $\boldsymbol{\beta}$ at the instant $k$ such that $\rho^\mFormula(\boldsymbol{\beta}, k) \geq 0$ if and only if $\boldsymbol{\beta} \mSat_k \mFormula$. We define such function recursively, as follows:
\begin{itemize}
    \item $\rho^{\pi^\mathbb{Q}}(\boldsymbol{\beta}, k) = \infty$ if $q_k \in \mathbb{Q}$, otherwise, $-\infty$;
    \item $\rho^{\pi_\epsilon^\mu}(\boldsymbol{\beta}, k) = -\boldsymbol{h}^\intercal\hat{\boldsymbol{x}}_k - c -\Phi^{-1}(1-\epsilon)\sqrt{\boldsymbol{h}^\intercal \Sigma_k^x \boldsymbol{h}}$;
    \item $\rho^{\mNot \pi}(\boldsymbol{\beta}, k) = -\rho^{\pi}(\boldsymbol{\beta}, k)$;
    \item $\rho^{\mFormula_1 \mAnd \mFormula_2}(\boldsymbol{\beta}, k) = \min\big(\rho^{\mFormula_1}(\boldsymbol{\beta}, k), \rho^{\mFormula_2}(\boldsymbol{\beta}, k)\big)$;
    \item $\rho^{\mFormula_1 \mOr \mFormula_2}(\boldsymbol{\beta}, k) = \max\big(\rho^{\mFormula_1}(\boldsymbol{\beta}, k), \rho^{\mFormula_2}(\boldsymbol{\beta}, k)\big)$;
    \item $\rho^{\mFormula_1 \mUntil_{[a, b]} \mFormula_2}(\boldsymbol{\beta}, k) = \max_{k^\prime\in[k+a,k+b]}\big( \min(\rho^{\mFormula_2}(\boldsymbol{\beta}, k^\prime), \\ \hspace*{4.5cm} \min_{k^{\prime\prime} \in [k, k^\prime]} \rho^{\mFormula_1}(\boldsymbol{\beta}, k^{\prime\prime})\big)$;
    \item $\rho^{\mAlways_{[a, b]} \mFormula}(\boldsymbol{\beta}, k) = \min_{k^\prime \in [k + a, k + b]} \rho^{\mFormula_2}(\boldsymbol{\beta}, k^\prime)$.
\end{itemize}

We can interpret the value of the function $\rho^\mFormula(\boldsymbol{\beta}, k)$ as \textit{how much} the trajectory $\boldsymbol{\beta}$ satisfies the formula $\mFormula$. The larger $\rho^\mFormula(\boldsymbol{\beta}, k)$ is, the further the trajectory is from violating the specification. 
%Note that $\rho^\mFormula(\boldsymbol{\beta}, 0) \geq 0$ is equivalent to $\boldsymbol{\beta} \mSat \mFormula$.

\subsection{Problem Formulation}

We can now formulate the active perception and motion planning problem for the switched linear control system in terms of a \gls*{prstl} specification. Active perception is inherently built into the framework because we plan in the belief space. In other words, to achieve the specification, the system may need to take actions that decrease uncertainty. 

\begin{problem}\label{prob:pr_1}
Given a switched linear control system, an initial condition $\bar{\boldsymbol{X}} \sim \mathcal{N}(\bar{\boldsymbol{x}}, \bar{\Sigma}^x)$, and a \gls*{prstl} formula $\mFormula$, design a control signal $\boldsymbol{a} = (q_0, \boldsymbol{u}_0)(q_1, \boldsymbol{u}_1)\dots$ that satisfies the formula $\mFormula$.  
\end{problem}

\section{Overall Approach}\label{sec:overall_approach}

Our proposed approach to solve \cref{prob:pr_1} is illustrated in \cref{fig:diag1}. Our basic idea is to construct deterministic abstractions and to use \textit{counterexample-guided synthesis} \cite{alur2013syntax,reynolds2015counterexample} to satisfy both the \gls*{prstl} specification $\mFormula$ and the dynamics of System (\ref{eq:prsystem}). Two interacting layers, discrete and continuous, work together to overcome nonconvexities in the logical specification $\mFormula$ efficiently. At the discrete layer, a discrete planner acts as a \textit{proposer}, generating discrete plans by solving a \gls*{bmc} for the given specification. We use an iterative deepening search to search first for shorter satisfying plans, thus minimizing undue computation. We pass the satisfying discrete plans to the continuous layer, which acts as a \textit{teacher}. In the continuous layer, a sampling-based search is applied to check whether a discrete plan is feasible. If the feasibility test does not pass, we construct a counterexample to discard infeasible trajectories. Then we add this counterexample to the discrete planner and repeat this process until we find a solution or no more satisfying plans exist.

  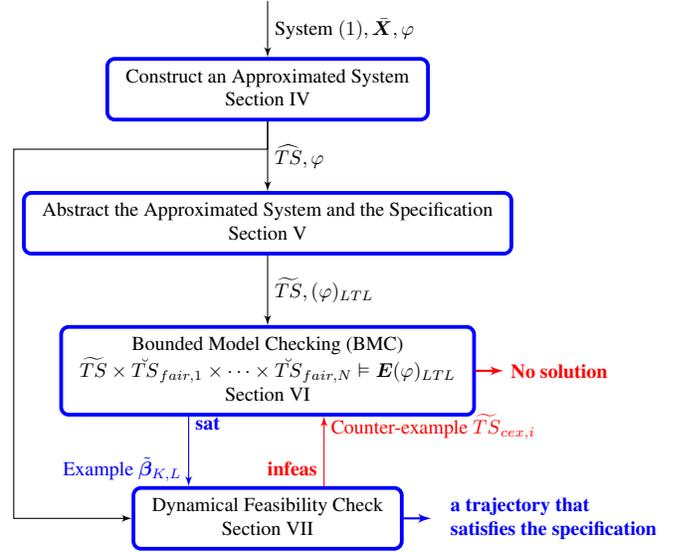
\begin{figure}
    \tikzstyle{block} = [draw, rounded corners=1mm, color=blue, text=black, line width=0.5mm, rectangle, minimum height=3em, minimum width=6em]
    \centering
    \begin{tikzpicture}[auto, >=latex', scale=0.75, transform shape]
        \node[block, minimum width=11em] (det) {\begin{tabular}{c}Construct an Approximated System \\ \cref{sec:det}\end{tabular}};
        \draw[->] ([yshift=1cm] det.north) -- node [pos=0.5, right] {$\text{System } (\ref{eq:prsystem}), \bar{\boldsymbol{X}}, \mFormula$} (det.north);
        \node[block, minimum width=11em,below=1.25cm of det] (abs) {\begin{tabular}{c}Abstract the Approximated System and the Specification \\ \cref{sec:abs}\end{tabular}};
        \draw[->] (det.south) -- node [pos=0.5, right] {$\widehat{TS}, \mFormula$} (abs.north);
        \node[block, minimum width=11em,below=1.25cm of abs] (dplan) {\begin{tabular}{c}Bounded Model Checking (BMC) \\ $\widetilde{TS} \times \breve{TS}_{fair,1} \times \dots \times \breve{TS}_{fair,N} \mSat \boldsymbol{E} (\mFormula)_{LTL}$ \\ \cref{sec:bmc}
        \end{tabular}};
        \draw[->] (abs.south) -- node [pos=0.5, right] {$\widetilde{TS}, (\mFormula)_{LTL}$} (dplan.north);
        \node[block, minimum width=11em,below=1.25cm of dplan] (reachsearch) {\begin{tabular}{c}Dynamical Feasibility Check \\ \cref{sec:feas} \end{tabular} };
        \draw[->] (det.south) |- ++(-4.5cm,-0.5cm) |- (reachsearch);
        \draw[->,color=blue] ([xshift=-14mm] dplan.south) -- node[left,near end] {\color{blue} Example $\tilde{\boldsymbol{\beta}}_{K,L}$} node[right,pos=0.1] {\color{blue} \textbf{sat}} ([xshift=-14mm] reachsearch.north);
        \draw[->,color=red] ([xshift=10mm]reachsearch.north) -- node[right,pos=0.85] {\color{red} Counter-example $\widetilde{TS}_{cex,i}$} node[left,near start] {\color{red} \textbf{infeas}} ([xshift=10mm] dplan.south);
        \draw[->,thick,color=blue] (reachsearch.east) -- node[right, pos=1] {\color{blue} \begin{tabular}{l}
             \textbf{a trajectory that} \\
             \textbf{satisfies the specification} 
        \end{tabular}} ++(0.5cm,0);
        \draw[->,thick,color=red] (dplan.east) -- node[right,pos=1] {\color{red} \textbf{No solution}} ++(0.5cm,0);
    \end{tikzpicture}
    \caption{Pictorial representation of our proposed approach. }
    \label{fig:diag1}
    \vspace{-0.5cm}
 \end{figure}

\section{Approximated Belief Dynamics}\label{sec:det}

Our approach to handle unknown observations during the planning is by approximating the belief system in \cref{eq:beliefsys} using \gls*{mlo}. However, we also need to abstract this approximated system to use \gls*{bmc}. Thus, we show that there is an equivalent transition system that models the resulting belief dynamics.

An initialized transition system $TS$ is defined by a tuple $(S, \bar{s}, Act, \delta)$, where:
\begin{itemize}
    \item $S$ is a set of states;
    \item $\bar{s} \in S$ is the initial state;
    \item $Act$ is a set of actions;
    \item $\delta : S \times Act \rightarrow S$ is a transition relation.
\end{itemize}

For example, a transition system $TS$ that models System \cref{eq:prsystem} is a transition system where: $s \in \mathbb{R}^{n(n+1)}$ is the vector defining the multivariate Gaussian belief state $\boldsymbol{X} \sim \mathcal{N}(\hat{\boldsymbol{x}}, \Sigma^x)$ mean $\hat{\boldsymbol{x}} \in \mathbb{R}^n$ and covariance $\Sigma^x \in \mathbb{R}^{n \times n}$, and $(q_k, \boldsymbol{u}_k, \boldsymbol{y}_{k + 1}) \in Act$ is an action defining  a system command $q_k$, input $\boldsymbol{u}_k$, and output $\boldsymbol{y}_{k + 1}$. Hence, $S \subseteq \mathbb{R}^{n(n+1)}$,  $\bar{s} = \bar{\boldsymbol{X}}$,  $Act \subseteq Q \times \mathcal{U} \times \mathbb{R}^p$, and $\delta$ is defined in \cref{eq:beliefsys}.

\subsection{Maximum Likelihood Observations}

The planning problem involves synthesizing a control signal $(q_0, \boldsymbol{u}_0) (q_1, \boldsymbol{u}_1) \dots$ which results in a belief trajectory $\boldsymbol{\beta}$ that satisfies the specification $\mFormula$. However, the belief dynamics also depend on the observations $\boldsymbol{y}_k$, which are available only after execution. Each observation signal $\boldsymbol{y} = \boldsymbol{y}_0\boldsymbol{y}_1 \dots$ results in a different belief trajectory $\boldsymbol{\beta}$ for the same initial condition and control signal. Therefore, we need to approximate the belief trajectory during the planning.

%Several works approximate the belief trajectory during the planning. Some works focus on choosing a particular observation signal during the planning under some assumption such as \gls*{mlo} \cite{platt2010belief}. Others consider the observations as random variables such as the generalized belief state \cite{indelman2016towards} and covariance steering approach \cite{ridderhof2020chance,zheng2021belief}. However, the approximation is required to include the active perception during the planning, and, in this work, we use the \gls*{mlo}.  

First, note that the observation during the planning is indeed a Gaussian random process: % in System (\ref{eq:prsystem}). However, applying it directly in the belief update turns the system unobservable. As a result, the planning will not do active perception anymore. We can compute the output random process as follows:
\begin{equation}
    \begin{aligned}
    \boldsymbol{Y}_{k+1} \sim  &  \int_{\boldsymbol{x}} P(\boldsymbol{y}_{k+1}|\boldsymbol{x},\boldsymbol{u}_k, q_k)P(\boldsymbol{x}|\boldsymbol{X}_k,\boldsymbol{u}_k, q_k) \\
\sim & \int_{\boldsymbol{x}}  \mathcal{N}\Big(\boldsymbol{y}_{k+1}| C_{q_k}\boldsymbol{x}, V_{q_k}V_{q_k}^\intercal \Big) \\
& \hspace{2.5cm}\mathcal{N}(\boldsymbol{x}|A_{q_k} \hat{\boldsymbol{x}}_k + B_{q_k} \boldsymbol{u}_k,\Sigma_{q_k,k}^{x,+}) \\
\sim & \mathcal{N}(C_{q_k}(A_{q_k} \hat{\boldsymbol{x}}_k + B_{q_k} \boldsymbol{u}_k),\Sigma_{q_k,k}^{z,+}).
\end{aligned}
\end{equation}
However, applying this directly in the belief update turns the system unobservable:
%Applying to \cref{eq:beliefsys}, we obtain:
\begin{equation}\label{eq:beliefsys_withoutperception}
\begin{aligned}
\hat{\boldsymbol{x}}_{k+1} = & (I_n - L_{q_k, k} C_{q_k}) (A_{q_k} \hat{\boldsymbol{x}}_k + B_{q_k} \boldsymbol{u}_k) \\
& \hspace{2cm}+ L_{q_k, k} C_{q_k} (A_{q_k} \hat{\boldsymbol{x}}_k + B_{q_k} \boldsymbol{u}_k) \\
= & A_{q_k} \hat{\boldsymbol{x}}_k + B_{q_k} \boldsymbol{u}_k, \\
\Sigma_{k+1}^x = & (I_n - L_{q_k, k} C_{q_k}) \Sigma_{q_k, k}^{x, +} + L_{q_k, k} \Sigma_{q_k, k}^{z,+} L_{q_k, k}^\intercal \\
 = & {\scriptstyle (I_n - L_{q_k, k} C_{q_k}) \Sigma_{q_k, k}^{x,+} + L_{q_k, k} \Sigma_{q_k, k}^{z+} (\Sigma_{q_k, k}^{z,+})^{-1} C_{q_k} \Sigma_{q_k, k}^{x,+}} \\
= & \Sigma_{q_k, k}^{x,+} = A_{q_k} \Sigma_k^x A_{q_k}^\intercal  + W_{q_k} W_{q_k}^\intercal.
\end{aligned}
\end{equation}
As a result, the uncertainty would never decrease, making active perception impossible. On the other hand, if we take the maximum likelihood observation assumption, the observation is always close to the values that it will turn to during the execution with active perception. 
%Hence, intuitively, if the trajectory with maximum likelihood cannot satisfy the specification, only unexpected events could allow the system to execute the task correctly. Although these events are possible, they are unlikely considering the current belief. Thus, we can use over-the-shelf feedback control techniques such as \gls*{lqr} or \gls*{lqg} \cite{antsaklis2007linear} to adapt while still satisfying the specification. 

\begin{example}
    In \cref{fig:light-dark-example-workspace}, the blue belief trajectory is a trajectory with MLO. The observation for this specific run is shown in purple. Note that the observation is very different from its maximum likelihood. However, we used \gls*{lqr} control to correct the estimated belief state (orange) towards the mean value of the planned belief trajectory. As a result, we kept the estimated belief state close to the planned one. Even though the state (in red) is initially different from the maximum likelihood, this strategy also forced the state trajectory to follow the specifications.
\end{example}

Applying the maximum likelihood observation $\boldsymbol{y}_{ml,k+1} = \argmax_{\boldsymbol{y}_{k+1}} P(\boldsymbol{y}_{k+1}|\boldsymbol{X}_k,\boldsymbol{u}_k, q_k) = C_{q_k}(A_{q_k} \hat{\boldsymbol{x}}_k + B_{q_k} \boldsymbol{u}_k)$ to the belief dynamics, we obtain the following approximated switched nonlinear control system,
\begin{equation}\label{eq:beliefsyssimple}
\begin{aligned}
\hat{\boldsymbol{x}}_{ml,k+1} = &
A_{q_k} \hat{\boldsymbol{x}}_{ml,k} + B_{q_k} \boldsymbol{u}_k, \\
\Sigma_{ml,k+1}^x = & (I_n - L_{q_k, k} C_{q_k}) \Sigma_{q_k, k}^{x,+}.
\end{aligned}
\end{equation}
The covariance update of the Kalman Filter is preserved using this approximation. Thus, minimizing the uncertainty during the planning can also minimize it during the execution.

\begin{example}
    \cref{fig:light-dark-system-maximum-likelihood} shows 100 Monte Carlo simulations for the light-dark example. The red regions are $95\%$ confidence regions of the simulated belief state, while the blue ones are the maximum likelihood trajectory. Note that the simulated trajectories are around the maximum likelihood. We can treat the error between the maximum likelihood belief trajectory and the simulated ones as disturbances. Thus, we could use a tracking strategy to deal with this error during execution.
    \begin{figure}
        \centering
        \includegraphics[width=0.4\textwidth]{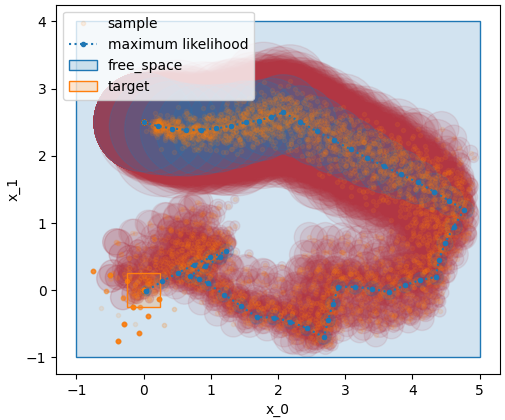}
        \caption{Monte Carlo simulation of 100 executions of a belief trajectory planned for the light-dark domain. The red regions and orange dots represent the sampled belief trajectories. The blue dot and regions represent the approximated belief trajectory.}
        \label{fig:light-dark-system-maximum-likelihood}
        \vspace{-0.5cm}
    \end{figure}
\end{example}

\subsection{Approximated Belief Transition System}

Similarly to the belief system, we can model the approximated belief system as a transition system $\widehat{TS}$. This transition system is a tuple $(\hat{S}, \bar{\hat{s}}, \widehat{Act}, \hat{\delta})$, where $\hat{S} = S$, $\bar{\hat{s}} = \bar{s}$, $\widehat{Act} \subseteq Q \times \mathcal{U}$, and $\hat{\delta}$ is defined in \cref{eq:beliefsyssimple}. Note that the approximated belief transition system $\widehat{TS}$ is a deterministic system that represent an underlying probabilistic switched linear control system. Thus, the approximated belief trajectory $\hat{\boldsymbol{\beta}}$ is a sequence $\boldsymbol{X}_0 \xrightarrow{q_0, \boldsymbol{u}_0} \boldsymbol{X}_1 \dots$.  We now can directly apply the \gls*{prstl} semantics to the approximate belief trajectories because it is independent of the value of observations. 

\section{Simulation Abstraction}\label{sec:abs}

The primary challenge to compute a plan for the system is the non-convex nature of the logical constraints. These non-convexities are, in general, hard to solve for trajectory synthesis approaches. We address this issue by proposing an abstraction that over-approximates the belief system in \cref{eq:beliefsyssimple}. The basic idea is to propose an abstraction that simulates the belief system. Its trajectories define a sequence of adjacent regions in the belief state parameters that guarantee the specifications. A simulation abstraction ensures that there is a solution only if there is a solution for the abstraction. Therefore, with an appropriate abstraction, we can search for all possible solutions more quickly. We define the labeling function $L : \hat{S} \rightarrow \pi_\epsilon^\mu \in cl(\mFormula)$ as the set of atomic probabilistic predicates in the closure of the \gls*{prstl} formula $\mFormula$ that a state $s$ satisfies. 

Now, we can consider the simulation relation between two transition systems \cite{baier2008principles}.
\begin{definition}\label{def:simulation}
    A relation $\mathcal{R} \subseteq \hat{S} \times \tilde{S}$ is a \textit{simulation relation}, i.e., the transition system $\widetilde{TS}$ simulates $\widehat{TS}$, if
    \begin{enumerate}
        \item $(\hat{\bar{s}}, \tilde{\bar{s}} ) \in \mathcal{R}$; \label{en:simulation_cond1}
        \item for every $( \hat{s}, \tilde{s} ) \in \mathcal{R}$ we have $L(\hat{s}) = L(\tilde{s})$; \label{en:simulation_cond2}
        \item for every $( \hat{s}, \tilde{s} ) \in \mathcal{R}$, if there exists $\hat{a} \in Act$ such that $\hat{s}^\prime = \delta(\hat{s}, \hat{a})$, there also exists a $\tilde{a} \in \widetilde{Act}$ such that $\tilde{s}^\prime = \tilde{\delta}(\tilde{s}, \tilde{a})$ and $(\hat{s}^\prime, \tilde{s}^\prime ) \in \mathcal{R}$. \label{en:simulation_cond3}
    \end{enumerate} 
\end{definition}

Since the probabilistic predicates form a second-order cone in the belief output parameter space, we can form a conic set $\mathcal{B} \subseteq \mathbb{R}^{n(n+1)}$ with the second-order cones formed by the set of predicates in the \gls*{prstl} formula. We can form a finite set of conic sets $\mathcal{B}$ with different combinations of probabilistic predicates $\pi_\epsilon^\mu \in cl(\mFormula)$ in the formula closure. As a result, we can build an abstraction with finite states using these partitions and define the transition between two partitions as valid if a control input takes any state in one partition to another. 

\begin{proposition}\label{prop:abs}
Given a \gls*{prstl} formula $\mFormula$, there always exists an abstraction transition system $\widetilde{TS} = (\tilde{S}, \bar{\tilde{s}}, \widetilde{Act}, \tilde{\delta})$ that simulates a transition system $\widehat{TS}$ describing the belief dynamical system \cref{eq:beliefsyssimple}.
\end{proposition}
\begin{proof}
We prove by construction. First, we define an abstraction state $\tilde{s} \in \tilde{S}$ for every different conic partition in the parameter space $\tilde{s} = \{ \bigwedge_{\pi_\epsilon^\mu \in \phi} \boldsymbol{X}_k \mSat \pi_\epsilon^\mu \} \subseteq  \mathbb{R}^{n(n+1)}$ for every state formula $\phi \in L(\tilde{s})$ in the closure. Note that for atoms true $\top$ and false $\perp$ we have $\tilde{s} = \mathbb{R}^{n(n+1)}$ and $\tilde{s} = \emptyset$, respectively. This definition enforce the condition \ref{en:simulation_cond2} in \cref{def:simulation}. We let the initial state of the abstraction be the state $\tilde{s}_0 \in \tilde{S}$ that represent the partition that include the initial belief state (i.e., $\hat{s}_0 \in \tilde{s}_0$). As a result, we guarantee the condition \ref{en:simulation_cond1} in \cref{def:simulation}. Now, the set of actions $\widetilde{Act}$ of the abstraction system $\widetilde{TS}$ is the set of system commands (i.e., $\widetilde{Act} = Q$).  Finally, we add transitions $\tilde{s}^\prime = \tilde{\delta}(\tilde{s}, q)$ for every two abstraction states $\tilde{s}, \tilde{s}^\prime \in \tilde{S}$ and action $q \in Q$ if and only if there is a control input $\boldsymbol{u} \in \mathcal{U}$ such that a transition $\hat{s}^\prime = \delta\big(\hat{s}, (q, \boldsymbol{u})\big)$ is valid in the transition system $\widehat{TS}$ and the states are in the partitions $\hat{s} \in \tilde{s}$ and $\hat{s}^\prime \in \tilde{s}^\prime$. Therefore, we satisfy the last condition (i.e., condition \ref{en:simulation_cond3} in \cref{def:simulation}) and guarantee the existence of an abstraction for any transition system $\widehat{TS}$ and \gls*{prstl} formula $\mFormula$.
\end{proof}

 \begin{example}
     Consider the specification of the light-dark domain shown in \cref{ex:light-dark-specs}. This specification generates two conic regions in the belief state parameters: 
     \begin{equation}
        \begin{aligned}
            \tilde{s}_1 = \{  \hat{x}_1-2.3\sigma_{1,1} & \geq -1, &
            \hat{x}_2-2.3\sigma_{2,2} & \geq -1,\\
            \hat{x}_1+2.3\sigma_{1,1} & \leq 5, & 
            \hat{x}_2+2.3\sigma_{2,2} & \leq 4 \}, \\
            \tilde{s}_2 = \{  \hat{x}_1-1.6\sigma_{1,1} & \geq -0.25, &
            \hat{x}_2-1.6\sigma_{2,2} & \geq -0.25, \\
            \hat{x}_1+1.6\sigma_{1,1} & \leq 0.25, &
            \hat{x}_2+1.6\sigma_{2,2} & \leq 0.25 \}.
        \end{aligned}
     \end{equation}
     These belief cones are illustrated in \cref{fig:light-dark-belief-cones} for covariances with $\sigma_{1,1} = \sigma_{2,2}$. As a result, the transition system $\widetilde{TS}$ that simulates the system is defined by two states, $\tilde{S} = \{ \tilde{s}_1, \tilde{s}_2 \}$, one action, $\widetilde{Act} = Q = \{1 \}$, and a transition function $\tilde{\delta}$ such that $\tilde{s}^\prime = \tilde{\delta}(\tilde{s}, \tilde{a})$ iff $\tilde{s}, \tilde{s}^\prime \in \tilde{S}$ and $\tilde{a} \in \widetilde{Act}$.
     \begin{figure}
        \centering
        \includegraphics[width=0.4\textwidth]{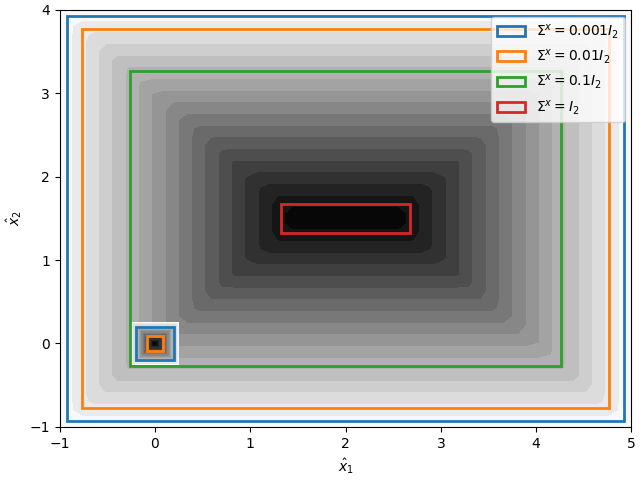}
        \caption{Illustrative representation of belief conic regions $\tilde{s}_1 \subset \mathbb{R}^6$ and $\tilde{s}_2 \subset \mathbb{R}^6$ for specification in \cref{ex:light-dark-specs} projected in the mean space ($(\hat{x}_1, \hat{x}_2) \in \mathbb{R}^2$) at different covariance values. Each covariance forms a valid polytopic region in the mean space. Note that the target region (small box at origin) has no valid mean value for covariances $\Sigma^x = 0.1 I_2$ and $I_2$. }
        \label{fig:light-dark-belief-cones}
     \end{figure}
 \end{example}

\begin{example}
In \cref{ex:laser-grasp-specs}, the state formulas describe five different belief cones: \textit{safe}, \textit{approach}, \textit{grasp}, \textit{move}, and \textit{place down}. These state formulas also include information about the system command/mode. Thus, we know that we are allowed to go to a belief state in the \textit{safe} and \textit{approach} only with system command \textit{opened} (i.e., $q = 1$), \textit{grasp} and \textit{place down} with $q = 2$, and \textit{move} with $q = 3$. We easily see that no valid speed $\|\boldsymbol{v} \|_{\infty} \leq 1$ can move from \textit{safe} to \textit{grasp} because the minimum distance between them is greater than one. Similarly, we cannot change the orientation fast enough (i.e., $|\omega| \leq 45^o$) to switch between \textit{grasp} and \textit{place down}. Therefore, the resulting abstraction is a transition system defined by transitions shown in \cref{tab:laser-grasp-transitions}.
\begin{table}[]
\caption{Abstraction Transitions for the Robot Manipulator}
\label{tab:laser-grasp-transitions}
\begin{tabular}{|lc|c|c|c|}
\hline
\multicolumn{2}{|c|}{\multirow{3}{*}{$\tilde{\delta}$}} & \multicolumn{3}{c|}{System Mode $q$}                                                            \\ \cline{3-5} 
\multicolumn{2}{|c|}{}                          & opened                         & closing/opening                & holding           \\
\multicolumn{2}{|c|}{}                          & 1                              & 2                              & 3                 \\ \hline
safe                  & $\tilde{s}_1$           & $\{\tilde{s}_1, \tilde{s}_2\}$ & $\emptyset$                    & $\emptyset$       \\ \hline
approach              & $\tilde{s}_2$           & $\{\tilde{s}_1, \tilde{s}_2\}$ & $\{\tilde{s}_3\}$              & $\{\tilde{s}_4\}$ \\ \hline
grasp                 & $\tilde{s}_3$           & $\{\tilde{s}_2\}$              & $\{\tilde{s}_3\}$              & $\{\tilde{s}_4\}$ \\ \hline
move                  & $\tilde{s}_4$           & $\{\tilde{s}_2\}$              & $\{\tilde{s}_3, \tilde{s}_5\}$ & $\{\tilde{s}_4\}$ \\ \hline
place down            & $\tilde{s}_5$           & $\{\tilde{s}_2\}$              & $\{\tilde{s}_5\}$              & $\{\tilde{s}_4\}$ \\ \hline
\end{tabular}
\end{table}
\end{example}

\subsection{Formula Transformation}\label{sec:ltltransf}

A deterministic and finite-state transition system allows us to implement many available model checking techniques such as \gls*{bmc} \cite{biere2006linear}. However, these techniques are designed for \gls*{ltl} formulas. Therefore, we need to transform \gls*{prstl} into \gls*{ltl} formulas.

We do this by taking an overapproximation of the delays and deadlines ($a,b$). We do this for several reasons. First, these temporal constraints increase the formula complexity by increasing the closure length. For example, a delay of ten instants would require ten nested next \gls*{ltl} operators. Additionally, later, we use a feasibility search that explores equivalent longer abstraction trajectories. Then, a shorter abstraction trajectory that satisfies the formula without time constraints will also represent a longer one that satisfies them in the feasibility search. For example, a one-step abstraction trajectory can satisfy the formula if we drop the delay of ten instants. We can use the quantitative semantics to check the time constraints during the feasibility search as we search for dynamical feasible belief trajectories. As the approach name suggests, if we can check shorter trajectories, \gls*{bmc} techniques require less computational effort.

Specifically, if we drop the delays and deadlines in the \gls*{prstl} formula, we obtain an \gls*{ltl} formula without the next operator. Thus, we define the following transformation:
\begin{equation}
\begin{aligned}
(\mFormula)_{LTL} = & \mFormula^\prime : \forall \mFormula_1 \circ \mFormula_2 \in cl(\mFormula), \\
& \begin{cases}
\mFormula_1 \circ \mFormula_2 & \textbf{if } \circ \in \{ \mAnd, \mOr \}, \\
\mFormula_1 \mUntil \mFormula_2 & \textbf{if } \circ = \mUntil_{[a,b]}  \textbf{ and } a = 0, \\
\mEventually \mFormula_1 \mUntil \mFormula_2 & \textbf{if } \circ = \mUntil_{[a,b]}  \textbf{ and } a > 0, \\
\mAlways \mFormula & \textbf{if } \circ = \mAlways_{[a,b]} \textbf{ and } a = 0, \\
\mEventually \mAlways \mFormula & \textbf{if } \circ = \mAlways_{[a,b]} \textbf{ and } a > 0,
\end{cases}
\end{aligned}
\end{equation}
where the \gls*{ltl} without next semantics is equal to \gls*{prstl} with all temporal operators having $a = 0$ and $b = \infty$.

Our goal is to guarantee that there is an abstraction trajectory $\tilde{\boldsymbol{\beta}}$ that satisfies the \gls*{prstl} formula $\mFormula$ only if there is one that satisfy the \gls*{ltl} formula $(\mFormula)_{LTL}$. Then we can use standard LTL model checking to propose candidate abstraction trajectories and decide if there are no new abstraction trajectories that are candidates for the problem.

\begin{proposition}\label{prop:ltl_trans}
An abstraction trajectory $\tilde{\boldsymbol{\beta}}$ satisfies a \gls*{prstl} formula $\mFormula$ \textbf{only if} this trajectory satisfies the \gls*{ltl} formula $(\mFormula)_{LTL}$.
\end{proposition}
\begin{proof}
We prove by contradiction. Assume that an abstraction trajectory $\tilde{\boldsymbol{\beta}}$ satisfies a \gls*{prstl} formula $\mFormula$ but does not satisfy the \gls*{ltl} formula $(\mFormula)_{LTL}$. By definition, the \gls*{ltl} semantics is equal to \gls*{prstl} with all temporal operators having $a = 0$ and $b = \infty$. From \gls*{prstl} semantics, if there is a trajectory that satisfies the \gls*{prstl} formula $\mFormula$ with $a > 0$ or $b < \infty$, it must satisfy with $a = 0$ and $b = \infty$, which contradicts the assumption. Therefore, the proposition follows.
\end{proof}

 \begin{example}
     We can easily abstract the formula that specifies the light-dark domain task (\cref{ex:light-dark-specs}) by just removing the time constraints in the until and always operators:
    \begin{equation}
     (safe \boldsymbol{U}_{[0,240]} \mAlways_{[0,10]} target)_{LTL} = safe \boldsymbol{U} \mAlways target,
     \end{equation}
     where $safe = \pi_{0.01}^{-x_1 - 1} \mAnd \pi_{0.01}^{x_1 - 5} \mAnd \pi_{0.01}^{-x_2 - 1} \mAnd \pi_{0.01}^{x_2 - 4}$ and $target = \pi_{0.05}^{-x_1 - 0.25} \mAnd \pi_{0.05}^{x_1 - 0.25} \mAnd \pi_{0.05}^{-x_2 - 0.25} \mAnd \pi_{0.05}^{x_2 - 0.25}$.     
    \end{example}

\section{Bounded Model Checking}\label{sec:bmc}

Now, we leverage the fact that the proposed abstraction $\widetilde{TS}$ is a finite-state transition system to propose an efficient method to check if it eventually satisfies the specifications. One prominent formal verification technique for these transition systems is the existential \gls*{smt} encoding of \gls*{ltl} \gls*{bmc} presented in \cite{rodriguesdasilva2021automatic} based on the work in \cite{biere2006linear}. As we showed in \cref{sec:ltltransf}, we can abstract  \gls*{prstl} into \gls*{ltl} formulas. Therefore, we can implement an \gls*{ltl} \gls*{bmc} to verify the existence of abstraction trajectories that satisfy the transformed formula. 

As its name suggests, a \gls*{bmc} encoding checks bounded trajectories. We can represent the infinite behavior of a linear temporal logic formula with a bounded trajectory with a loop (i.e., $(K, L)$-trajectories). In \cref{fig:bmc_flowchart}, we illustrate this algorithm as a flowchart. The algorithm starts with one-step trajectories and increases the trajectory bound if no trajectory is satisfiable. This search ends without a solution when, for a bound $K$, there is no satisfying loop-free trajectory. Intuitively, if a bad state is reachable, it is reachable via a trajectory with no duplicate states, i.e., a loop-free trajectory. The soundness and completeness of the encoding used in this paper are proved in \cite[Proposition 3]{rodriguesdasilva2021automatic} and \cite[Proposition 4]{rodriguesdasilva2021automatic}. We refer the interested reader to \cite{rodriguesdasilva2021automatic} for further details.

  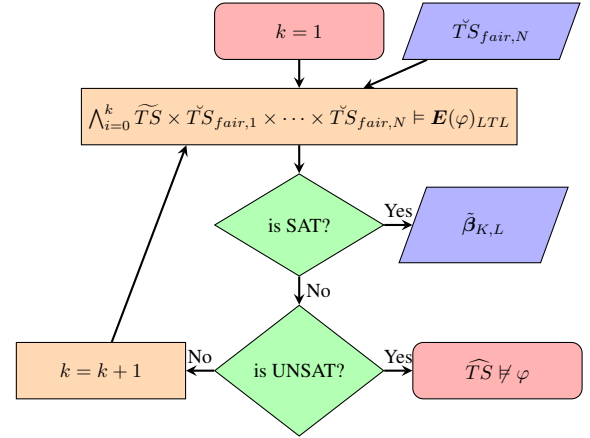
\begin{figure}
    \tikzstyle{startstop} = [rectangle, rounded corners, minimum width=3cm, minimum height=1cm,text centered, draw=black, fill=red!30]
    \tikzstyle{io} = [trapezium, trapezium left angle=70, trapezium right angle=110, minimum width=3cm, minimum height=1cm, text centered, draw=black, fill=blue!30]
    \tikzstyle{process} = [rectangle, minimum width=3cm, minimum height=1cm, text centered, draw=black, fill=orange!30]
    \tikzstyle{decision} = [diamond, minimum width=3cm, minimum height=1cm, text centered, draw=black, fill=green!30]
    \tikzstyle{arrow} = [thick,->,>=stealth]
    \centering
    \begin{tikzpicture}[auto, >=latex', scale=0.75, transform shape]
        \node[startstop] (ini) {$k=1$};
        \node[process, below=0.5cm of ini] (assert) {$\bigwedge_{i=0}^k\widetilde{TS} \times \breve{TS}_{fair, 1} \times \dots \times \breve{TS}_{fair, N} \mSat \boldsymbol{E} (\mFormula)_{LTL}$};
        \node[decision, below=0.5cm of assert] (issat) {is SAT?};
        \node[decision, below=0.5cm of issat] (isunsat) {is UNSAT?};
        \node[process, left=0.5cm of isunsat] (increment) {$k = k + 1$};
        \node[io, right=0.5cm of issat] (rho) {$\tilde{\boldsymbol{\beta}}_{K,L}$};
        \node[startstop, right=0.5cm of isunsat] (exit) {$\widehat{TS} \not\mSat \mFormula$};
        \node[io, right=0.5cm of ini] (cex) {$\breve{TS}_{fair, N}$};
        \draw[arrow] (ini) -- (assert);
        \draw[arrow] (cex) -- (assert);
        \draw[arrow] (assert) -- (issat);
        \draw[arrow] (issat) -- node[above] {Yes} (rho);
        \draw[arrow] (issat) -- node[right] {No} (isunsat);
        \draw[arrow] (isunsat) -- node[above] {Yes} (exit);
        \draw[arrow] (isunsat) -- node[above] {No} (increment);
        \draw[arrow] (increment) -- ([xshift=-2cm] assert.south);
    \end{tikzpicture}
    \caption{Flowchart representation of the \gls*{bmc} algorithm for \gls*{cegis} presented in \cite{rodriguesdasilva2021automatic}. }
    \label{fig:bmc_flowchart}
    \vspace{-0.5cm}
 \end{figure}   
 
 As we show later, the counterexample represents the candidate trajectories (or prefixes) as transition systems. This transition system discard not only the satisfying trajectory $\tilde{\boldsymbol{\beta}}_{K,L}$ (or its prefix) but also all extended trajectories generated by loops inside the trajectory. Therefore, the search is complete and finishes in a finite time.

 \begin{example}\label{ex:light-dark-bmc}
     In the light-dark domain, the first abstraction trajectory that the proposed \gls*{bmc} returns is always $\tilde{\boldsymbol{\beta}}_{2,2} = \tilde{s}_1 \xrightarrow{1} \tilde{s}_2 (\xrightarrow{1} \tilde{s}_2)^\omega$, where $\tilde{s}_1$ and $\tilde{s}_2$ are illustrated in \cref{fig:light-dark-belief-cones}. 
 \end{example}

\section{Dynamical Feasibility Check}\label{sec:feas}

Although the proposed simulation abstraction eases the computation of the existential verification, it is not enough to ensure that the abstraction $\widetilde{TS}$ existentially satisfies the specification. As discussed above, the satisfaction $\tilde{\boldsymbol{\beta}}_{K,L} \mSat (\mFormula)_{LTL}$ of the \gls*{ltl} transformation $(\mFormula)_{LTL}$ does not ensure that there exists a belief trajectory $\hat{\boldsymbol{\beta}}_{K,L}$ that satisfies the \gls*{prstl} formula $\mFormula$. A naive way to handle this fact is to check every abstraction trajectory $\tilde{\boldsymbol{\beta}}_{K,L}$ for the existence of a trajectory $\hat{\boldsymbol{\beta}}_{K,L}$ with same bound $K$ and loop $L$ and discard them if there is not. However, this approach would require many calls to the \gls*{bmc} solver and causing additional computational effort. Therefore, we propose another abstraction that generates \gls*{ltl} equivalent abstraction trajectories without calling the \gls*{bmc} solver, reducing the overall computational effort. 

We propose a variant of SPARSE-RRT \cite{littlefield2013efficient} for active perception that guarantees \gls*{ltl} equivalence. Intuitively, the abstraction trajectory creates a tunnel inside the belief state parameters (i.e., mean and covariance) that simulates the satisfaction of the specification. Our RRT variant searches for belief trajectories inside this tunnel by driving the mean values toward the end (towards achieving the task) and minimizing the uncertainty (active perception).

We need an RRT variant for active perception because sampling-based methods for deterministic systems do not solve active perception efficiently. There are several reasons for this. First, the belief dynamics in \cref{eq:beliefsyssimple} increase the system dimension exponentially (i.e., $n(n + 1)$), meaning that the problem is more complicated to solve even for simple systems. Second, the task and active perception planning is a dual objective problem not addressed in typical sampling-based methods. The goal cannot only reach the desired target (or targets) position. The uncertainty is not controllable. Thus, inputs that locally drive the trajectory towards the target are often not the inputs that will drive globally. Another reason is that over-the-shelf sampling-based techniques do not consider trajectories with a loop. Therefore, we propose modifications to RRT that take these characteristics into account. 

We also proposed an RRT variant for active perception in \cite{da2019active}. However, there are critical differences between the SPARSE-RRT variant proposed in this work and the one proposed in our previous work. First, in \cite{da2019active}, we maximized a robustness value to search for dynamically feasible belief trajectories. However, if two trajectories end with mean values close to each other, the one with less uncertainty is more likely to be a prefix of a satisfying trajectory than the one that is more robust. Another difference is in the computation of the control signal. In this work, we propose a \gls*{lp} problem to include the tunnel constraints during the control synthesis. This \gls*{lp} problem is inspired by our recent works for deterministic systems \cite{rodriguesdasilva2021automatic,da2021symbolic}.

\subsection{Stable Sparse Rapidly-Exploring Random Trees}

Now, we detail the main steps of the proposed SPARSE-RRT variant. This approach is illustrated in \cref{alg:feas}. An feasibility tree $(\mathbb{V}, \mathbb{E})$ is a directed graph, where $\mathbb{V}$ is a set of nodes, and $\mathbb{E}$ is a labeled set of edges. Each edge is labed with a tuple $(q, \boldsymbol{u})$, where $q$ is a system mode (command) and $\boldsymbol{u}$ is a system input. Each node is also a tuple $(\tilde{k}, \hat{s})$, where $\hat{s} \in \hat{S}$ is an approximated belief state such that is simulated by an instant $\tilde{k}$ in the satisfying abstraction trajectory (i.e., $(\tilde{s}_{\tilde{k}}, \hat{s}) \in \mathcal{R}$ and $\tilde{\boldsymbol{\beta}}_{K,L} = \tilde{s}_0\dots\tilde{s}_{\tilde{k}} \dots \tilde{s}_K$ such that $\tilde{\boldsymbol{\beta}}_{K,L} \vDash (\mFormula)_{LTL}$). As a result, a path of this tree represents both a prefix of an abstraction and an approximated belief trajectory.

\begin{algorithm}
	\caption{$feasibility\_search$}\label{alg:feas}
    \begin{algorithmic}[1]
		\REQUIRE $\tilde{\boldsymbol{\beta}}_{K,L}, \bar{\boldsymbol{X}}, \mathcal{U}, N, \Delta_{drain}, \Delta_{near}, bias, h_{lb}, h_{ub}$
		    \STATE $\hat{\boldsymbol{\beta}}^*_{K^*, L^*} = \emptyset$, $\mathbb{V}_{active} = \emptyset$, $\mathbb{E} = \emptyset$;
		    \STATE $\mathbb{V}_{active} \gets (0, \bar{s})$
			\FOR{$j = 0$ \TO $N$} 	
			    \STATE $v_{rand} = random\_sample(\mathbb{V}_{active})$
			    \STATE $\boldsymbol{X}_{rand} \sim \mathcal{N}\Big(unif\_sample\big(\mathcal{P}(v_{rand}.\tilde{s})\big), 0\Big)$
			    \STATE $v_{near} \gets best\_nearest(\boldsymbol{X}_{rand}, \mathbb{V}_{active}, \Delta_{near})$
                \STATE $v_{new} \gets propagate( v_{near}, \mathbb{E}, bias, h_{lb}, h_{ub}, \tilde{\boldsymbol{\beta}}_{K,L}, \mathcal{U})$
                \IF{$\rho^{\mFormula}(\hat{\boldsymbol{\beta}}_{K_{new}, L_{new}}^{v_{new}}, 0) \geq 0$}                    \IF{$\rho^{\mFormula}(\hat{\boldsymbol{\beta}}_{K_{new}, L_{new}}^{v_{new}}, 0) > \rho^{\mFormula}(\hat{\boldsymbol{\beta}}^*_{K^*, L^*}, 0)$}
                        \STATE $\hat{\boldsymbol{\beta}}^*_{K^*, L^*} = \hat{\boldsymbol{\beta}}_{K_{new}, L_{new}}^{v_{new}}$
                    \ENDIF
                \ELSE
		            \STATE $\mathbb{V}_{active} \gets  v_{new}$
		            \STATE $drain(\mathbb{V}_{active}, v_{new}, \Delta_{drain})$
		        \ENDIF
			\ENDFOR
			\IF{$\hat{\boldsymbol{\beta}}^*_{K^*, L^*} = \emptyset$}
			    \STATE $i = \max_{v_j \in \mathcal{T}.\mathbb{V}_{active}} \max_{k \in [0, K]} k$ s.t. $\tilde{s}_{K_j}^{v_j} = \tilde{s}_k$	
			    \IF{$i < K$}
			        \RETURN $\breve{TS}_{fair}(Prefix_i(\tilde{\boldsymbol{\beta}}_{K,L}))$
			     \ELSE
			        \RETURN $\breve{TS}_{fair}(\tilde{\boldsymbol{\beta}}_{K,L})$
		        \ENDIF
		    \ELSE
		        \RETURN $\hat{\boldsymbol{\beta}}^*_{K^*, L^*}$
		    \ENDIF
\end{algorithmic}
\end{algorithm}

We only keep track of a sparse tree formed by the active nodes $\mathbb{V}_{active}$. We start this tree with the root (line 2). Then we start the main loop. In summary, the loop procedure has four main steps. First, we randomly sample a node giving more chances to less uncertain ones (lines 4-6). Note that $\mathcal{P}(\tilde{s}) = \cup_{\forall \pi_\epsilon^{\boldsymbol{h}^\intercal \boldsymbol{x} + c} \in L(\tilde{s})} \{\boldsymbol{h}^\intercal \hat{\boldsymbol{x}}_h + c  \leq 0\}  \subseteq \mathbb{R}^n$ is a polytopic region in the mean space for the abstraction state $\tilde{s}$. Second, we propagate this node (line 7), synthesizing a finite horizon control. Later, since we dropped the time constraints from the specification to use \gls*{bmc}, we check if this new node satisfies the \gls*{prstl} formula and save the one with the larger robustness (line 8-10). Finally, if the node does not satisfy the specification, we add it to the active nodes and drain near nodes with larger uncertainty to keep the active nodes sparse. We repeat this loop for $N$ times. If no solution is found, we return a counterexample $\breve{TS}_{fair}$ instead (lines 14-21).

In the following section, we explore this procedure in more detail.

\subsection{Metrics}

Our aim is to define the distance and cost functions to maximize the exploration of the state space and minimize uncertainty. We can explore the state space by guiding the mean values of the belief state. Thus, we define the distance $m(v, \boldsymbol{X})$ between a node $v = (\tilde{k}, \boldsymbol{X} \sim \mathcal{N}(\hat{\boldsymbol{x}}, \Sigma^x))$ and a belief state $\boldsymbol{X} \sim \mathcal{N}(\hat{\boldsymbol{x}}, \Sigma^x)$ as the Euclidean distance between the mean vector:
\begin{equation}\label{eq:feas:meas}
    m(v, \boldsymbol{X}) := \| v.\hat{\boldsymbol{x}} - \hat{\boldsymbol{x}} \|_2.
\end{equation}

As a result, we search a set of nodes to find those that are near or a single node that is the nearest as follows:
\begin{align}
    near(\mathbb{V}, \boldsymbol{X}, \Delta) = & \{ v \in \mathbb{V} : m(v, \boldsymbol{X}) \leq \Delta \}, \\
    nearest(\mathbb{V}, \boldsymbol{X}) = & \argmin_{v \in \mathbb{V}} m(v, \boldsymbol{X}).
\end{align}

Since the distance focuses on exploring the state space, the cost function focuses on uncertainty minimization. Note that the uncertainty constraint can be heterogeneous in the chance constraints (i.e., the tolerance can be different for different atomic predicates). Hence, we measure the current uncertainty using the atomic predicates. We can measure the uncertainty with respect to an atomic predicate $\pi_\epsilon^{\boldsymbol{h}^\intercal \boldsymbol{x} + c}$ at an instant $k$ using the expression $\sqrt{\boldsymbol{h}^\intercal \Sigma_k^x \boldsymbol{h}}$. In other words, the larger the uncertainty for a predicate $\pi_\epsilon^{\boldsymbol{h}^\intercal \boldsymbol{x} + c}$, the larger the measure $\sqrt{\boldsymbol{h}^\intercal \Sigma_k^x \boldsymbol{h}}$ is. Hence, to measure the uncertainty of a node $v$, we take into consideration the latest abstraction and belief states associated with this node:
\begin{equation}\label{eq:feas:cost}
    \begin{aligned}
        c(v) = & \max_{\forall \pi_\epsilon^{\boldsymbol{h}^\intercal \boldsymbol{x} + c} \in L(\tilde{s}_{v.\tilde{k}})}  \Phi^{-1}(1 - \epsilon) \sqrt{\boldsymbol{h}^\intercal v.\Sigma^{x} \boldsymbol{h}}.
    \end{aligned}
\end{equation}
This cost is not additive as required in the original SPARSE-RRT. The reason is that we are not minimizing the typical goal of reaching a target here. Thus, a longer trajectory is not intrinsically worst for the active perception goal.

\begin{example}
    Consider the sampling strategy in the light-dark domain illustrated in \cref{fig:light-dark-example-bestnearest} with near distance $\Delta_{near} = 1$. First, we randomly select the index $i_{rand} = 0$ and the mean $\hat{\boldsymbol{x}}_{rand} = [1.5, 1]^\intercal$, the red star. In this case, the best nearest node is the one associated with the trajectory in blue. Next, if we select the $\hat{\boldsymbol{x}}_{rand} = [3, 1]^\intercal$, no node is withing the near distance. Thus, the best nearest is the node associated with trajectory that finishes nearest to this position, i.e., the orange trajectory.

    \begin{figure}
        \centering
        \includegraphics[width=0.9\linewidth]{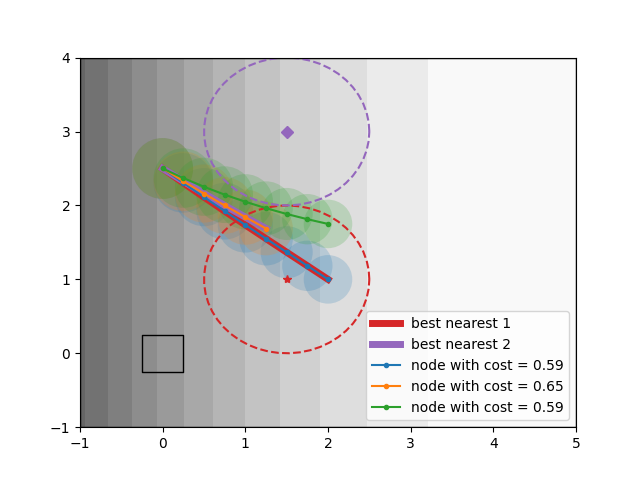}
        \caption{Illustration of the random sampling strategy in the light-dark domain.}
        \label{fig:light-dark-example-bestnearest}
        \vspace{-0.5cm}
    \end{figure}
\end{example}

\begin{example}
    Consider the new node generate by propagating the the blue trajectory in \cref{fig:light-dark-example-bestnearest}. If the drain distance $\Delta_{drain} = 0.5$, the node associated with the green trajectory is drained, as shown in \cref{fig:light-dark-example-drain}.
 
    \begin{figure}
        \centering
        \includegraphics[width=0.9\linewidth]{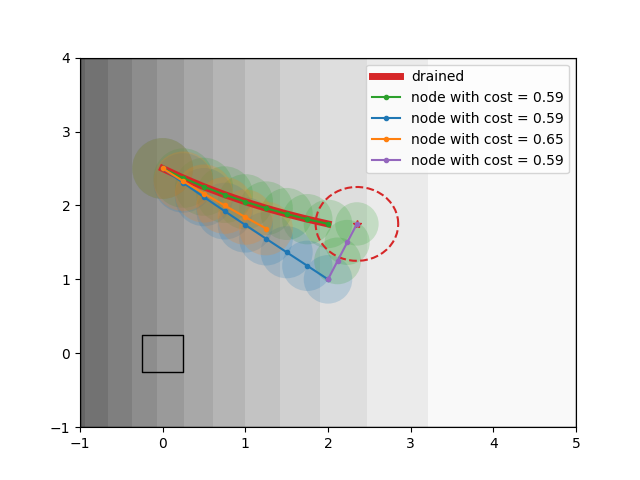}
        \caption{Illustration of the propagation and drain strategy in the light-dark domain.}
        \label{fig:light-dark-example-drain}
        \vspace{-0.5cm}
    \end{figure}
\end{example}

\subsection{Propagate}

The propagation algorithm is responsible for generating nodes from the best nearest nodes. The basic idea is a two-step synthesis. First, we use only the belief mean dynamics, which is linear, in a \gls*{lp} problem. Next, we check for controls that drive the belief inside the target tunnel. Therefore, we explore the state space by driving the mean of the belief state. 

The propagation can be a random or a biased walk. In the random walk, we drive the mean towards a random target $\check{\hat{\boldsymbol{x}}} = unif\_sample\big(\mathcal{P}(\tilde{s}_{v.\tilde{k}})\big) \in \mathbb{R}^n$. This step allows us to explore the belief space for actions that reduce uncertainty (active perception). In the biased walk, we drive the belief mean towards a region $\mathcal{P}(\tilde{s}_{v.\tilde{k}+1})$ in the next step in the tunnel. Hence, we search for actions towards achieving the task. 

The control synthesis for a random walk is the following \gls*{lp} problem:
\begin{equation}\label{eq:propagate:lp}
\begin{aligned}
     \min\limits_{\substack{\hat{\boldsymbol{x}}_0,\dots, \hat{\boldsymbol{x}}_h \in \mathbb{R}^n, \\
     \lambda_0, \dots, \lambda_h \in \mathbb{R}, \\
     \boldsymbol{u}_0, \dots, \boldsymbol{u}_{h-1} \in \mathbb{R}^m}} \hspace{0.1cm} &  \lambda_h \\
     \mathrm{s.t. } & \hat{\boldsymbol{x}}_0 = v_{near}.\hat{\boldsymbol{x}}, \| \check{\hat{\boldsymbol{x}}} - \hat{\boldsymbol{x}}_h \|_{\infty} \leq \lambda_h, \\
    & \hat{\boldsymbol{x}}_{k} = A_{q_{v.\tilde{k}- 1}} \hat{\boldsymbol{x}}_{k-1} + B_{q_{v.\tilde{k}- 1}} \boldsymbol{u}_{k-1},  \\
    &  \boldsymbol{h}^\intercal \hat{\boldsymbol{x}}_{k-1} + c  \leq \| \boldsymbol{h} \| \lambda_k, \forall \pi_\epsilon^{\boldsymbol{h}^\intercal \boldsymbol{x} + c} \in L(\tilde{s}_{v.\tilde{k}}),  \\
    & \frac{\bar{\lambda}}{\delta} \lambda_{k-1} \leq \lambda_k, \boldsymbol{u}_{k-1} \in \mathcal{U}, \\
    & k = 1, \dots, h.
\end{aligned}
\end{equation}
Intuitively, this problem propagates a parent node $v_{near} \in \mathbb{V}_{active}$ belief mean $v_{near}.\hat{\boldsymbol{x}}$ towards a target belief mean $\check{\hat{\boldsymbol{x}}} \in \mathbb{R}^n$ for a finite horizon $h \in [h_{lp}, h_{ub}]$ that is randomly sampled. This problem minimizes the Chebyshev distance to the target belief mean (i.e., $\| \check{\hat{\boldsymbol{x}}} - \hat{\boldsymbol{x}}_{h} \|_{\infty}$) while satisfying the trajectory segment $(\xrightarrow{q_{v.\tilde{k} - 1}} \tilde{s}_{v.\tilde{k} })^{h-1}$.

We use a slack technique to ensure that it is always feasible and returns the largest dynamically feasible prefix. In other words, if it is not dynamically feasible to generate a trajectory segment $(\xrightarrow{q_{v.\tilde{k} - 1}} \tilde{s}_{v.\tilde{k} })^{h-1}$, the first slack $\lambda_k$ with value greater than the tolerance $\delta$ give us the largest dynamically feasible horizon $h^\ast = k - 1$. The reason is that $\bar{\lambda}$ is maximum distance of a point to a half-space defined by the predicates:
\begin{equation}
    \begin{aligned}
         \max_{\hat{\boldsymbol{x}} \in \mathbb{R}^n} \max_{\pi_{\epsilon^\prime}^{\boldsymbol{h}^{\prime\intercal} \boldsymbol{x} + c^\prime} \in L(\tilde{s}_{v.\tilde{k}})} & \boldsymbol{h}^{\prime\intercal} \hat{\boldsymbol{x}} + c^\prime \\
         \mathrm{s.t. } & \boldsymbol{h}^\intercal \hat{\boldsymbol{x}} + c \leq 0, \forall \pi_\epsilon^{\boldsymbol{h}^\intercal \boldsymbol{x} + c} \in L(\tilde{s}_{v.\tilde{k}}).
    \end{aligned}
\end{equation}
Thus, the constraint $\frac{\bar{\lambda}}{\delta} \lambda_k \leq \lambda_{k+1}$ ensure that if there is no control signal that satisfy the prefix $(\xrightarrow{q_{v.\tilde{k} - 1}} \tilde{s}_{v.\tilde{k} })^{h-1}$, the minimal of the problem is a trajectory segment with the largest dynamically feasible prefix.

With $bias$ probability, instead of a random walk, we implement a biased walk. The control synthesis for a biased walk is the following \gls*{lp} problem:
\begin{equation}\label{eq:propagate:lp:biased}
\begin{aligned}
     \min\limits_{\substack{\hat{\boldsymbol{x}}_0,\dots, \hat{\boldsymbol{x}}_h \in \mathbb{R}^n, \\
     \lambda_0, \dots, \lambda_h \in \mathbb{R}, \\
     \boldsymbol{u}_0, \dots, \boldsymbol{u}_{h-1} \in \mathbb{R}^m}} \hspace{0.1cm} &  \lambda_h \\
     \mathrm{s.t. } & \hat{\boldsymbol{x}}_0 = v_{near}.\hat{\boldsymbol{x}}, \\
    & \hat{\boldsymbol{x}}_{h} = A_{q_{v.\tilde{k}}} \hat{\boldsymbol{x}}_{h-1} + B_{q_{v.\tilde{k}}} \boldsymbol{u}_{h-1},  \\
    &  \boldsymbol{h}^\intercal \hat{\boldsymbol{x}}_h + c  \leq \| \boldsymbol{h} \| \lambda_k, \forall \pi_\epsilon^{\boldsymbol{h}^\intercal \boldsymbol{x} + c} \in L(\tilde{s}_{v.\tilde{k} + 1}),  \\
    & \hat{\boldsymbol{x}}_{i} = A_{q_{v.\tilde{k}- 1}} \hat{\boldsymbol{x}}_{i-1} + B_{q_{v.\tilde{k}- 1}} \boldsymbol{u}_{i-1},  \\
    &  \boldsymbol{h}^\intercal \hat{\boldsymbol{x}}_{k-1} + c  \leq \| \boldsymbol{h} \| \lambda_k, \forall \pi_\epsilon^{\boldsymbol{h}^\intercal \boldsymbol{x} + c} \in L(\tilde{s}_{v.\tilde{k}}),  \\
    & \frac{\bar{\lambda}}{\delta} \lambda_{k-1} \leq \lambda_k, \boldsymbol{u}_{k-1} \in \mathcal{U}, \\
    & i = 1, \dots, h-1, k = 1, \dots, h.
\end{aligned}
\end{equation}
In plain English, we propagate a parent node $v_{near} \in \mathbb{V}_{active}$ belief mean $v_{near}.\hat{\boldsymbol{x}}$ towards a target polytopic region in the mean space $\mathcal{P}(\tilde{s}_{v.\tilde{k} + 1})$ for a random horizon $h \in [h_{lp}, h_{ub}]$. This problem minimizes the distance to the Chebyshev center of the tunnel $(\xrightarrow{q_{v.\tilde{k} - 1}} \tilde{s}_{v.\tilde{k}})^{h-1}\xrightarrow{q_{v.\tilde{k}}} \tilde{s}_{v.\tilde{k} + 1}$. Thus, as secondary goal, this problem also maximizes the \gls*{prstl} robustness in this segment.

 \begin{example}
    Consider that we selected the node with the blue trajectory in \cref{fig:light-dark-example-bestnearest}, target mean value is $\check{\hat{\boldsymbol{x}}} = [2.35, 1.75]^\intercal$, and the horizon $h = 3$. The trajectory associated with the new node is illustrated in \cref{fig:light-dark-example-drain} as the concatenation of blue and purple trajectories. 
 \end{example}

\subsection{Belief Trajectory with a Loop}

We establish that we can use an approximated belief system to solve \cref{prob:pr_1}. However, temporal logic specifications may require infinite behavior, and designing an arbitrary infinite control trajectory is not computationally tractable. Thus, we represent these infinite behaviors as trajectories with a loop.
 \begin{definition}\label{def:kltrajectory}
     Given a \gls*{prstl} formula $\mFormula$, an initial condition $\bar{\boldsymbol{X}}$, and a control signal $\boldsymbol{a} = (q_0, \boldsymbol{u}_0)\dots(q_k, \boldsymbol{u}_K)$, a belief system trajectory $\hat{\boldsymbol{\beta}}_{K, L}$ with a loop, or simply $(K,L)$-trajectory, is a sequence of the form:
     \begin{equation}
         \hat{\boldsymbol{\beta}}_{K, L} := \boldsymbol{X}_0 \xrightarrow{q_0, \boldsymbol{u}_0} \dots \boldsymbol{X}_{L-1}  \Big(\xrightarrow{q_{L-1}, \boldsymbol{u}_{L-1}} \boldsymbol{X}_L \dots \boldsymbol{X}_K\Big)^\omega,
     \end{equation}
     where $\hat{\boldsymbol{x}}_K = \hat{\boldsymbol{x}}_{L-1}$, and $\sqrt{\boldsymbol{h}^\intercal \Sigma_K^x \boldsymbol{h}} \leq \sqrt{\boldsymbol{h}^\intercal \Sigma_{L-1}^x \boldsymbol{h}}$ for all $\pi^{\mu}_\epsilon \in cl(\mFormula) : \boldsymbol{X}_{L-1} \mSat \pi^{\mu}_\epsilon$.
 \end{definition}

\begin{example}
    Consider the trajectory shown in \cref{fig:light-dark-example-workspace}. This trajectory can represent an infinite behavior by repeating a zero control input for any instant after nineteen. Hence, a $(K,L)$-loop trajectory $\boldsymbol{\beta}_{19,19}$ can be illustrated with the prefix $Prefix_{18}(\boldsymbol{\beta}_{19,19})$. The loop would be $(\xrightarrow{q_{18},\boldsymbol{u}_{18}} \boldsymbol{X}_{19})^\omega$, where $q_{18} = 1$, $\boldsymbol{u}_{18} = [0, 0]^\intercal$, and $\hat{\boldsymbol{x}}_{19} = \hat{\boldsymbol{x}}_{18} \approx [0, 2.15]^\intercal$. Since the system is observable, $\sqrt{\boldsymbol{h}^\intercal \Sigma_k^x \boldsymbol{h}} \leq \sqrt{\boldsymbol{h}^\intercal \Sigma_{18}^x \boldsymbol{h}}$ for all $\pi^{\mu}_\epsilon \in \{ \pi_{0.05}^{-x_1 - 0.25}, \pi_{0.05}^{x_1 - 0.25}, \pi_{0.05}^{-x_2 + 2}, \pi_{0.05}^{x_2 - 2.5} \}$ and $k > 18$. 
\end{example}

\subsection{PrSTL Semantics for Belief Trajectory with a Loop}

A belief $(K,L)$-trajectory does not have the covariance for instants larger than its bound (i.e., $K$). As a result, we cannot apply the \gls*{prstl} semantics in this trajectory without unrolling it. However, to guarantee completeness for unbounded temporal operators is not tractable because it requires us to unroll this trajectory infinite times. Therefore, we propose a tractable semantics for $(K,L)$-trajectories that ensure that $\hat{\boldsymbol{\beta}} \mSat \mFormula$ \textbf{if} $\hat{\boldsymbol{\beta}}_{K,L} \mSat \mFormula$, where $\hat{\boldsymbol{\beta}}$ is an infinite trajectory represented by the $(K,L)$-trajectory $\hat{\boldsymbol{\beta}}_{K,L}$. 

Formally, only the semantics of temporal operators are different from the original semantics presented in \cref{sec:prstl}:
\begin{itemize}
  \item $\hat{\boldsymbol{\beta}}_{K,L} \mSat_{k} \mFormula_1 \mUntil_{[a,b]} \mFormula_2$
    \begin{itemize}
        \item \textbf{if} $\boldsymbol{k + a \leq L}$ \textbf{and} $\exists k^\prime \in [k + a, \min(k+b, K)]$, $\boldsymbol{\beta}_{K,L} \mSat_{{k^\prime}}\mFormula_2$, and, $\forall k^{\prime\prime} \in [k + a, k^\prime]$, $\boldsymbol{\beta}_{K,L} \mSat_{{k^{\prime\prime}}}\mFormula_1$,
        \item \textbf{or if}, $Unroll(N_{K,L}(k + a), \boldsymbol{\beta}_{K,L}) \mSat_{k} \mFormula_1 \mUntil_{[a,b]} \mFormula_2$, ;
    \end{itemize}
    \item $\hat{\boldsymbol{\beta}}_{K,L} \mSat_{k} \mAlways_{[a,b]} \mFormula$ 
    \begin{itemize}
        \item \textbf{if} $\boldsymbol{k + a \leq L}$ \textbf{and}  $\forall  k^\prime \in [k + a, \min(k+b, K)]$, $\hat{\boldsymbol{\beta}}_{K,L} \mSat_{{k^\prime}}\mFormula_2$, 
        \item \textbf{or if}, $Unroll(N_{K,L}(k + a), \boldsymbol{\beta}_{K,L}) \mSat_{k} \mAlways_{[a,b]} \mFormula$,
    \end{itemize}
\end{itemize} 
where $N_{K,L}(k) = \max\left(\left\lceil\frac{k - L}{K + 1 - L}\right\rceil, 0\right)$, 
\begin{align*}
    Unroll(N, \boldsymbol{\beta}_{K,L}) =  \boldsymbol{X}_0 \dots \overbrace{\boldsymbol{X}_{L-1} \xrightarrow{q_{L-1}, \boldsymbol{u}_{L-1}} \boldsymbol{X}_L \dots \boldsymbol{X}_K \dots}^{N \text{ times}} \\
     \Big(\xrightarrow{q_{L^\prime-1}, \boldsymbol{u}_{L^\prime-1}} \boldsymbol{X}_{L^\prime} \dots \boldsymbol{X}_{K^\prime}\Big)^\omega,
\end{align*}
$L^\prime = L + N (K + 1 - L)$, and $K^\prime = K + N (K + 1 - L)$.

Intuitively, the bounded trajectory $\boldsymbol{X}_0 \dots \boldsymbol{X}_K$ is a conservative witness of the system infinite behavior of a $(K,L)$-trajectory $\hat{\boldsymbol{\beta}}_{K,L}$. By definition, the belief at the latest instant (i.e., $\boldsymbol{X}_K$) is less uncertain than the belief before the loop (i.e., $\boldsymbol{X}_{L-1}$). By repeating the same control inputs, the operator $Unroll$ conservatively approximates the uncertainty. Therefore, if this conservative trajectory satisfies the specification, the exact one surely satisfies too. 

\subsection{Quantitative Semantics for Trajectories with a Loop}

From the semantics for $(K,L)$-trajectory, we also obtain a real-valued function $\rho^\mFormula$ such that $\hat{\boldsymbol{\beta}} \mSat_k \mFormula$ \textbf{if} $\rho^\mFormula(\hat{\boldsymbol{\beta}}_{K,L}, k) \geq 0$. Similarly, only the quantitative semantics of temporal operators are redefined as follows:
\begin{align*}
&\rho^{\mFormula_1 \mUntil_{[a, b]} \mFormula_2}(\hat{\boldsymbol{\beta}}_{K,L}, k) = \\
&\begin{cases}
    \begin{aligned}
        \max_{k^\prime\in[k+a,\min(k+b, K)]}\big( \min(\rho^{\mFormula_2}(\hat{\boldsymbol{\beta}}_{K,L}, k^\prime), \\
        \min_{k^{\prime\prime} \in [k, k^\prime]} \rho^{\mFormula_1}(\hat{\boldsymbol{\beta}}_{K,L}, k^{\prime\prime})\big),
    \end{aligned}
    & \textbf{if }k + a \leq L \\
    \rho^{\mFormula_1 \mUntil_{[a, b]} \mFormula_2}\big(Unroll(N_{K,L}(k + a), \boldsymbol{\beta}_{K,L}), k\big), & \textbf{otherwise},
\end{cases} \\
&\rho^{\mAlways_{[a, b]} \mFormula_2}(\hat{\boldsymbol{\beta}}_{K,L}, k) = \\
&\begin{cases}
    \min_{k^\prime\in[k+a,\min(k+b, K)]}\big( \min(\rho^{\mFormula}(\hat{\boldsymbol{\beta}}_{K,L}, k^\prime)\big)
    & \textbf{if }k + a \leq L \\
    \rho^{\mFormula_1 \mUntil_{[a, b]} \mFormula_2}\big(Unroll(N_{K,L}(k + a), \boldsymbol{\beta}_{K,L}), k\big), & \textbf{otherwise}.
\end{cases}
\end{align*}
Since the uncertainty unrolling the loop using the operator $Unroll$ is larger than exact ones, the robustness measure $\rho^{\mFormula}(\hat{\boldsymbol{\beta}}_{K,L}, k)$ of a $(K,L)$-trajectory is a lower bound of the infinite trajectory $\hat{\boldsymbol{\beta}}$, i.e., $\rho^{\mFormula}(\hat{\boldsymbol{\beta}}_{K,L}, k) \leq \rho^{\mFormula}(\hat{\boldsymbol{\beta}}, k)$. Therefore, $\rho^{\mFormula}(\hat{\boldsymbol{\beta}}, k) \geq 0$ if $\rho^{\mFormula}(\hat{\boldsymbol{\beta}}_{K,L}, k) \geq 0$.

\subsection{Counter-examples}\label{sec:cex}

Since the abstraction simulates the approximated belief system, we may not find an approximated trajectory inside the tunnel formed by the abstraction trajectory. If this happens, we want to discard the abstraction trajectories in the feasibility tree in the future \gls*{bmc} executions. Our approach is to model these trajectories as a fair transition system.

An initialized fair transition system $TS_{fair}$ is defined by a tuple $(S, \bar{s}, Act, \delta, F)$, where $(S, \bar{s}, Act)$ is a transition system, and $F \subseteq S$ is a fairness condition. A trajectory of a fair transition system is valid only if starts at $\bar{s}$ and visit an state in $F$ infinitely often. 

Unlike the abstraction and approximated belief systems, the counter-example transition system actions are defined over Boolean formulas about transitions. A proposition about the abstraction trajectory transitions is a tuple $(\tilde{s}, q, \tilde{s}^\prime)$ such that an abstraction trajectory satisfies this proposition at instant $k$ ($\tilde{\beta} \mSat_k (\tilde{s}, q, \tilde{s}^\prime)$) if and only if $\tilde{s}_k = \tilde{s}$, $q_k = q$, and $\tilde{s}_{k+1} = \tilde{s}^\prime$. Now, we formally define the counter-example as a fair transition system:
\begin{definition}\label{def:cex}
    Given an abstraction trajectory $\tilde{\boldsymbol{\beta}}_{K,L}$, a counter-example is a fair transition system $\breve{TS}_{fair}(\tilde{\boldsymbol{\beta}}_{K,L})$ defined by the tuple $(\breve{S}, \bar{\breve{s}}, \breve{\delta}, \breve{Act}, \breve{F})$ where:
    \begin{itemize}
        \item $\breve{S} = \{ 0, \dots, K, K + 1\}$, $\bar{\breve{s}} = 0$, $\breve{F} = \{ K + 1 \}$,
        \item $\breve{Act}$ is a set of Boolean formulas over propositions about the abstraction trajectory transitions,
        \item the following transitions are valid:
        \begin{itemize}
            \item for $k = 0, \dots, K - 1$, 
            \begin{itemize}
                \item $k - 1 = \breve{\delta}\big(k, (\tilde{s}_k, q_k, \tilde{s}_{k-1}\big)$ if $k > 0$,
                \item $k + 1 = \breve{\delta}\big(k, (\tilde{s}_k, q_k, \tilde{s}_{k+1}\big)$,
                \item $k = \breve{\delta}\big(k, (\tilde{s}_k, q_k, \tilde{s}_{k})\big)$, and,
                \item $K + 1 = \breve{\delta}\big(k, \mNot (\tilde{s}_k, q_k, \tilde{s}_{k+1}) \mAnd \mNot (\tilde{s}_k, q_k, \tilde{s}_{k})\big)$;
            \end{itemize}
            \item $L = \breve{\delta}\big(K, (\tilde{s}_K, q_{L-1}, \tilde{s}_L\big)$;
            \item $K = \breve{\delta}\big(K, (\tilde{s}_K, q_{K-1}, \tilde{s}_K)\big)$; and,
            \item $K + 1 = \breve{\delta}\big(K, \mNot (\tilde{s}_K, q_{L-1}, \tilde{s}_L) \mAnd \mNot (\tilde{s}_K, q_{K-1}, \tilde{s}_K)\big)$,
            \item $K + 1 = \breve{\delta}(K + 1, \mTrue)$.
        \end{itemize}
    \end{itemize}
    
\end{definition}

The intuition is that any extension of the trajectory $\tilde{\boldsymbol{\beta}}_{K,L}$ generated by loops, i.e., generate by composition of the following transitions: $\tilde{s}_0 \rightarrow \tilde{s}_0$, $\tilde{s}_0 \xrightarrow{q_0} \tilde{s}_1$, $\dots$, $\tilde{s}_k \xrightarrow{q_{k-1}} \tilde{s}_{k-1}$, $\tilde{s}_k \xrightarrow{q_{k-1}} \tilde{s}_{k}$, $\tilde{s}_k \xrightarrow{q_k} \tilde{s}_{k+1}$, $\dots$, $\tilde{s}_K \xrightarrow{q_{K-1}} \tilde{s}_{K-1}$, $\tilde{s}_K \xrightarrow{q_{K-1}} \tilde{s}_{K}$, $\tilde{s}_K \xrightarrow{q_{L-1}} \tilde{s}_L$. Thus, if the trajectory do not take one of these transitions, the system goes to the counter-example state $K + 1$, which is the accepting state.

\section{Iterative-deepening Algorithm}\label{sec:idprstl}

The proposed iterative deepening search is illustrated in \cref{alg:idprstl}. A stop criterion in the \gls*{bmc} determines when it is not existentially satisfiable. This criterion is used as a conservative stop condition for iterative deepening the bounded checking. While we do not achieve that, if \gls*{bmc} find that the system existentially satisfies the \gls*{ltl} abstraction of the specification, we use \cref{alg:feas} to check if the example that proves this claim (an abstraction trajectory) is dynamically feasible. Otherwise, this algorithm returns a counter-example that we produced with the abstraction in the next \gls*{bmc} checking.

\begin{algorithm}
	\caption{$Iterative-deepening\_search$}\label{alg:idprstl}
    \begin{algorithmic}[1]
		\REQUIRE {\small System (\ref{eq:prsystem}), \mFormula, $, \bar{\boldsymbol{X}}, \mathcal{U}, N, \Delta_{drain}, \Delta_{near}, bias, h_{lb}, h_{ub}$}
		    \STATE $K = 1$
		    \STATE $\widehat{TS} $ from System (\ref{eq:prsystem})
		    \STATE $\widetilde{TS}$ from $\widehat{TS} $
		    \WHILE{\NOT $\widetilde{TS} \times \breve{TS}_{fair,1} \times \dots \times \breve{TS}_{fair,N} \not\mSat \boldsymbol{E} (\mFormula)_{LTL}$}
		    	\IF{$\widetilde{TS} \times \breve{TS}_{fair,1} \times \dots \times \breve{TS}_{fair,N} \mSat \boldsymbol{E} (\mFormula)_{LTL}$}
		    	    \IF{$feasibility\_search = \hat{\boldsymbol{\beta}}^\ast_{K^\ast, L^\ast}$}
		    	        \RETURN $\hat{\boldsymbol{\beta}}^\ast_{K^\ast, L^\ast}$
		    	    \ENDIF
		    	\ELSE
		    	    \STATE $K = K + 1$
		    	\ENDIF
		    \ENDWHILE
		    \RETURN unsat
\end{algorithmic}
\end{algorithm}

We show that \cref{alg:idprstl} only returns approximated belief trajectories that satisfy the specification.
\begin{theorem}\label{theo:feas}
\cref{alg:idprstl} returns a belief trajectory $\boldsymbol{\beta}_{K^\ast,L^\ast}^\ast$ \textbf{only if} this trajectory satisfies the \gls*{prstl} formula $\mFormula$.
\end{theorem}
\begin{proof}
In \cref{alg:feas}, we check the specification using the quantitative semantics and ensure that the trajectory $\boldsymbol{\beta}_{K^{new},L^{new}}^{new}$ also satisfy the \gls*{prstl} formula $\mFormula$, i.e., $\boldsymbol{\beta}_{K^{new},L^{new}}^{new} \mSat \mFormula$. Since this algorithm synthesizes inputs for the system in \cref{eq:beliefsyssimple}, we have that the control inputs in $\boldsymbol{\beta}_{K^{new},L^{new}}^{new} \mSat \mFormula$ is an approximated solution for \cref{prob:pr_1}.
\end{proof}

Since \gls*{bmc} is complete, if no solution exists, \cref{alg:idprstl} finishes in a finite time.
\begin{theorem}\label{theo:completness}
\cref{alg:idprstl} always terminates in a finite time.
\end{theorem}
\begin{proof}
First, from \cite[Proposition 4]{rodriguesdasilva2021automatic}, we have that \gls*{bmc} is complete. Thus, the abstraction $\widetilde{TS}$ does not existentially satisfies the formula $(\mFormula)_{LTL}$, and, from \cref{prop:abs,prop:ltl_trans}, there is no solution for \cref{prob:pr_1}. \cref{alg:feas} always finishes in a finite time because of the timeout parameter $N$. Finally, if there are satisfying trajectories in the abstraction but none is dynamically feasible, a finite number of counterexamples exclude them for two reasons. First, there are finite loop-free trajectories in a transition system with a finite number of states. Second, the counterexamples exclude a trajectory and all others generated by loops. Therefore,  \cref{alg:idprstl} must terminate in a finite time.
\end{proof}

\section{Experiments}\label{sec:experiments}

\subsection{Light-Dark Domain}
    As discussed in \cref{ex:light-dark-system,ex:light-dark-specs}, the light-dark domain is a robot motion planning problem that requires active perception. The robot must move to a target position, but the position measurement is noisy, and the noise level depends on the robot's current position. 
    
    \cref{fig:light-dark-example-workspace} shows a solution to the light-dark domain. The blue trajectory is the planned trajectory found by the proposed algorithm with feasibility timeout $t_{out} = 15s$, near distance $\Delta_{near} = 1$, drain distance $\Delta_{drain} = 0.5$, goal bias $bias = 0.25$, and horizon bound $h_{lb} = 3$ and $h_{ub} = 8$. We execute this plan using a \gls*{rhc} strategy with a finite horizon discrete time \gls*{lqr} to track the mean of the estimated belief with cost function $J = \hat{\boldsymbol{x}}_h^\intercal Q \hat{\boldsymbol{x}}_h + \sum_{k=0}^{h - 1} \hat{\boldsymbol{x}}_k^\intercal Q \hat{\boldsymbol{x}}_k + \boldsymbol{u}_k^\intercal R \boldsymbol{u}_k$, where $Q = I_2$ and $R = 0.05 I_2$ and the horizon $h = 5$. The red trajectory is the state trajectory, the purple trajectory is the output trajectory, and the orange is the resulting belief trajectory. Note that the output trajectory is noisy and cannot be used without considering its uncertainty. Even that the initial state is outside the region of $95\%$ of confidence, the tracking strategy was able to follow the planned belief trajectory. As a result, it accomplished the task. 
    
    We use Monte Carlo simulation to compute the probability that the actual trajectory generated during the execution satisfies the specifications. From $30$ different executions of \cref{alg:idprstl}, we obtained an average probability of $92.07\% \pm 9.21$. This probability is fairly high considering that it is the product that the belief at each instant satisfies all predicates satisfied in the approximated trajectory. In \cref{fig:light-dark-example-result}, we show one of these results for one hundred Monte Carlo simulations. Each red point is a state in one of these simulations. The result is that the probability of the plan satisfying the specification using the tracking strategy is $96\%$. This probability is within the probability in the specification (i.e., $95\%$ of confidence). We can see in \cref{fig:light-dark-example-result} that the actual trajectories (in red) that failed the specifications failed at achieving the target region. 
    \begin{figure}
        \centering
        \includegraphics[width=0.9\linewidth]{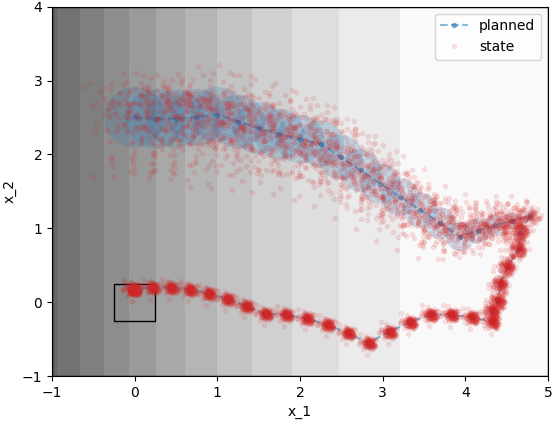}
        \caption{Illustrative representation of one hundred Monte Carlo simulations of tracking the planned trajectory in the light-dark domain.}
        \label{fig:light-dark-example-result}
        \vspace{-0.5cm}
    \end{figure}
    
    Although we approximate the observations during the planning, this experiment suggests that this approximation can capture essential behaviors of the plan execution. Additionally, the task was not satisfiable without the observations due to the initial uncertainty. Thus, the planning strategy had to combine not only task and motion planning but also active perception. 
    
\subsection{Robot Manipulation Domain}

    The robot manipulation domain is introduced in \cref{ex:laser-grasp-system,ex:laser-grasp-specs}. The robot must grasp and place a puck down in a target position. However, the puck and target position is not known perfectly. Therefore, the robot must also ``learn'' these positions during the execution of the plan (i.e., active perception). 
    
    Consider that the robot gripper starts at $(17.5, 17.5, \pi / 2)$ in the robot configuration space. In the same configuration space, the puck is initially at $(3.5, 12)$, but our initial belief is that it is at $\mathcal{N}\big((1, 10), I_2)\big)$. Similarly, the center of the target location is at $(6.5, 3)$, but our initial belief is that it is at $\mathcal{N}\big((10, 2), 2 I_2)\big)$. Note that the robot may not grasp the puck without active perception because of the initial error between the initial belief and the real puck position. Additionally, the robot could not place the puck down correctly in the target location even if it grasped it. Therefore, we need active perception.
    
    We generated a planned trajectory using the proposed algorithm with feasibility timeout $t_{out} = 60s$, near distance $\Delta_{near} = 5$, drain distance $\Delta_{drain} = 2.5$, goal bias $bias = 0.25$, and horizon bound $h_{lb} = 2$ and $h_{ub} = 5$. During execution this plan, we use a finite horizon \gls*{lqr} to track the planned belief state with parameters $Q = diag([1, 1, 0, 0, 1])$ for mode $q = 1$, $Q = diag([0, 0, 1, 1, 1])$ for mode $q = 3$, $R = 0.05 I_2$ and the horizon $h = 10$. The result of one execution is shown in \cref{fig:laser-grasp-solution}. 
    
\begin{figure}
    \centering
    \begin{subfigure}[b]{0.4\textwidth}
        \centering
        \includegraphics[width=\textwidth]{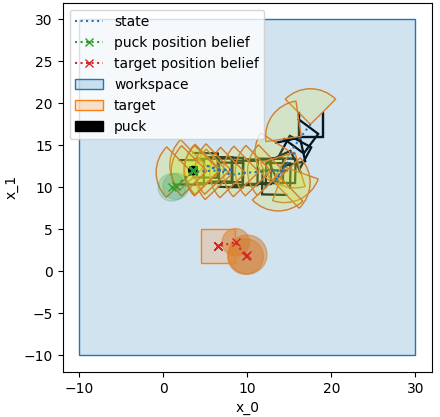}
        \caption{Before Grasping}
        \label{fig:laser-grasp-solution-before-grasping}
    \end{subfigure}
    \begin{subfigure}[b]{0.4\textwidth}
        \centering
        \includegraphics[width=\textwidth]{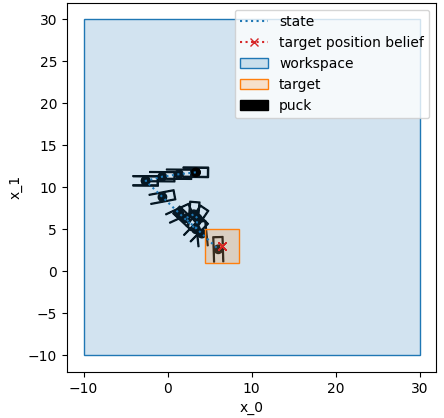}
        \caption{After Grasping}
        \label{fig:laser-grasp-solution-after-grasping}
    \end{subfigure}
    \caption{Illustrative representation of a solution for the robot manipulation domain. The trajectory is on the robot workspace, and the relative position of the puck and target belief position is translated to this workspace. The gripper is represented in black for each state position, and a yellow wedge represents the camera.}
    \label{fig:laser-grasp-solution}
\vspace{-0.5cm}
\end{figure}
    
    Note that the robot moves the gripper towards the belief target location before grasping the puck. Hence, the algorithm automatically identifies that the robot must learn the ``actual'' target position before holding the puck when the camera is blocked. By doing that, the robot can correctly execute the task even under incorrect initial belief. 

%The experime are implemented by our  used idSTLPy toolbox \cite{dasilva2022idstlpy} to run these experiments. This a open-source toolbox available at \url{https://codeocean.com/capsule/0013534/tree}.

\section{Conclusion and Future Work}\label{sec:conclusion}

We presented a framework for controller synthesis from PrSTL specifications for active perception tasks. This framework allows for active perception in complex tasks specified by PrSTL formulas. We demonstrated the efficacy of our approach on a simulation of robot motion planning and robot manipulation problems.  We also implemented our design methods into a toolbox, called idSTLPy \cite{dasilva2022idstlpy}, which is available at \texttt{\url{https://codeocean.com/capsule/0013534/tree}}.
In our future work, we plan to extend this results to other probabilistic hybrid systems and also consider probabilistic switching. Another direction is to replace the \gls*{mlo} approximation during the planning.

\bibliographystyle{IEEEtran}
\bibliography{IEEEabrv,library}

\begin{IEEEbiography}
    [{\includegraphics[width=1in,height=1.25in,clip,keepaspectratio]{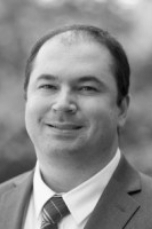}}]{Rafael Rodrigues da Silva} (M'16) received a bachelor in science (BS) degree in Electrical Engineering from the State University of Santa Catarina, in Joinville, Brazil, his hometown. Before going into academia, he had worked for almost ten years in the steel industry as a control engineer, where he lead a team of high qualified engineers and technicians in several projects. During this time, he received a master in science (MS) degree in Electrical Engineering and Industrial Computing at UTFPR, in Curitiba, Brazil, when he focused his research in high-performance Genetic Algorithms in FPGA (Field Programmable Gate Array) for Computer Vision. He is currently pursuing the Ph.D. degree in electrical engineering at University of Notre Dame, Notre Dame, IN, USA.
    
    His research interest focus on the design of intelligent physical systems, and, currently, he is working with verification and automatic synthesis of hybrid systems from high-level specifications.
\end{IEEEbiography}

\begin{IEEEbiography}
    [{\includegraphics[width=1in,height=1.25in,clip,keepaspectratio]{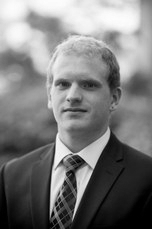}}]{Vince Kurtz}
    is a Dolores Z. Liebmann fellow and a PhD student in Electrical Engineering at the University of Notre Dame (Notre Dame Indiana, 46556). He studied physics at Goshen College (Goshen Indiana, 46526). His research interest in long-term autonomy for high degree-of-freedom robots lies at the intersection of robotics, control theory, and formal methods. 
\end{IEEEbiography}

\begin{IEEEbiography}
    [{\includegraphics[width=1in,height=1.25in,clip,keepaspectratio]{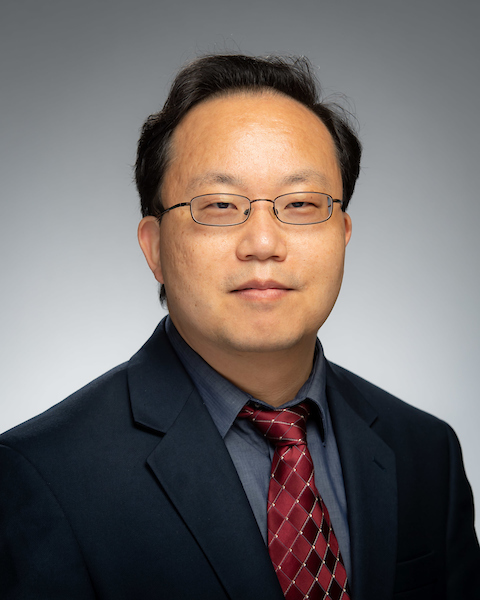}}]{Hai Lin} 
(SM’10) is currently a full professor at the Department of Electrical Engineering, University of Notre Dame, where he got his Ph.D. in 2005. Before returning to his {\em alma mater}, he has been working as an assistant professor in the National University of Singapore from 2006 to 2011. Dr. Lin's recent teaching and research activities focus on the multidisciplinary study of fundamental problems at the intersections of control theory, machine learning and formal methods. His current research thrust is motivated by challenges in cyber-physical systems, long-term autonomy, multi-robot cooperative tasking, and human-machine collaboration. Dr. Lin has been served in several committees and editorial board, including {\it IEEE Transactions on Automatic Control}. He served as the chair for the IEEE CSS Technical Committee on Discrete Event Systems from 2016 to 2018, the program chair for IEEE ICCA 2011, IEEE CIS 2011 and the chair for IEEE Systems, Man and Cybernetics Singapore Chapter for 2009 and 2010. He is a senior member of IEEE and a recipient of 2013 NSF CAREER award.
\end{IEEEbiography}

\end{document}